\documentclass[journal, a4paper, twoside]{IEEEtran}
\IEEEoverridecommandlockouts
\usepackage[utf8]{inputenc} 
\usepackage[T1]{fontenc}
\usepackage{url}
\usepackage{ifthen}
\usepackage{cite}
\usepackage[cmex10]{amsmath} 
\usepackage{amssymb}
\usepackage{graphicx}
\usepackage{epstopdf}
\usepackage{epsfig}
\usepackage{multirow}
\usepackage{amsthm}
\usepackage{comment}
\usepackage{caption}
\usepackage{enumitem}
\usepackage{amsmath}
\usepackage{color}
\usepackage{bbm}
\usepackage{algorithm}
\usepackage{algpseudocode}
\usepackage{csquotes}
\usepackage{subfiles}
\usepackage{caption, multirow, makecell}
\usepackage{array}
\usepackage{caption}
\usepackage{nicematrix}
\usepackage{cuted}
\usepackage{subcaption}
\usepackage{mathtools}
\usepackage[thinc]{esdiff}

\DeclarePairedDelimiter\ceil{\lceil}{\rceil}
\DeclarePairedDelimiter\floor{\lfloor}{\rfloor}

\newtheorem{thm}{Theorem}
\newtheorem{lem}{Lemma}

\newtheorem{example}{Example}

\newtheorem{corollary}{Corollary}

\newtheorem{rem}{Remark}

\allowdisplaybreaks

\def\BibTeX{{\rm B\kern-.05em{\sc i\kern-.025em b}\kern-.08em
    T\kern-.1667em\lower.7ex\hbox{E}\kern-.125emX}}

\begin{document}

\title{Coded Caching with Shared Caches and \\ Private Caches }

\author{\IEEEauthorblockN{ Elizabath Peter, K. K. Krishnan Namboodiri, and B. Sundar Rajan\\}
	\IEEEauthorblockA{Department of Electrical Communication Engineering, IISc
		Bangalore, India \\
		E-mail: \{elizabathp, krishnank, bsrajan\}@iisc.ac.in}
	\thanks{A part of the content of this manuscript appeared in \textit{Proc. IEEE Inf. Theory Workshop (ITW)}, 2023 \cite{PNR3}.}
}
\maketitle
\begin{abstract}
This work studies the coded caching problem in a setting where users can access a private cache of their own along with a shared cache. The setting consists of a server connected to a set of users, assisted by a smaller number of helper nodes that are equipped with their own storage. In addition to the helper caches, each user possesses a dedicated cache which is also used to prefetch file contents. Each helper cache can serve an arbitrary number of users, but each user gets served by only one helper cache. We consider two scenarios: (a) the server has no prior information about the user-to-helper cache association, and (b) the server knows the user-to-helper cache association at the placement phase itself. We design centralized coded caching schemes under uncoded placement for the above two settings. For case (b), we propose four schemes, and establish the optimality of certain schemes in specific memory regimes by deriving matching lower bounds. The fourth scheme, called as the composite scheme, appropriately partitions the file library and the cache memories between two other schemes in case (b) to minimize the rate by leveraging the advantages of both. 
\end{abstract}


\section{Introduction}
The proliferation of immersive applications has led to an exponential growth in multimedia traffic which in turn has brought the usage of caches distributed in the network to the fore. Coded caching, introduced by Maddah-Ali and Niesen \cite{MaN}, is a technique to combat the traffic congestion experienced during peak times by exploiting the caches in the network. Coded caching consists of two phases: \textit{placement (prefetching)} phase and \textit{delivery} phase. During the placement phase, the network is not congested, and the caches are populated with portions of file contents in coded or uncoded form, satisfying the memory constraint. The subsequent delivery phase commences with the users requesting files from the server. The server exploits the prefetched cache contents to create coded multicasting opportunities so that the users' requests are satisfied with a minimum number of bits sent over the shared link. By allowing coding in the delivery phase, a global caching gain proportional to the total memory size available in the network is achieved in addition to the local caching gain. The global caching gain is a multicasting gain available simultaneously for all possible demands. The network model considered in \cite{MaN} is that of a dedicated cache network with a server having access to a library of $N$ equal-length files connected to $K$ users through an error-free shared link. Each user possesses a cache of size equal to $M$ files. The coded caching scheme in \cite{MaN}, which is referred to as MaN scheme henceforth, was shown to be optimal when $N \geq K$ under the constraint of uncoded cache placement \cite{WTP, YMA}. The coded caching approach has then been studied in different settings such as decentralized placement \cite{MaN2}, % coded placement \cite{CFL,AmG,TiC,Gom1,ZhTi,Gom2},
 caching with shared caches \cite{PUE,IZW,PeR,PNR2}, caching with users accessing multiple caches \cite{HKD,MKR,BrE}, caching with security \cite{STC} and secrecy constraints \cite{RPKP} for the demanded file contents, and many more.

In this work, we consider a network model as shown in Fig.~\ref{fig:setting}. There is a server with $N$ equal-length files, connected to $K$ users and to $\Lambda \leq K$ helper caches through a shared bottleneck link. Each helper cache is capable of storing $M_s$ files. Unlike the shared cache networks discussed in \cite{PUE,IZW,PeR,PNR2}, each user is also endowed with a cache of size $M_p$ files. The conceived model represents scenarios in a wireless network such as a large number of cache-aided users (for example, mobile phones, laptops, and other smart devices) are connected to the main server via a set of access points which is generally less in number than the number of users. These access points are provided with storage capacity such that the users connected to them have access to the cache contents. In this setting, two scenarios can arise. One scenario is a content placement agnostic of which user will be connected to which helper cache, and the other is a situation where the server is aware of the association of users to helper caches during the placement. The information about user-to-helper cache association may not be available in all cases. For instance, the users can be mobile, hence, the helper node or access point to which the users are connected also change according to their location. However, in some cases, it is possible to obtain the location information of the users a priori (for example, if the users are static).  Therefore, it is important to consider both the scenarios.

 A network model similar to that in Fig.~\ref{fig:setting} was considered in \cite{MeR} and \cite{PNR} to account for the secrecy constraints in the conventional shared cache setup. However, in \cite{MeR} and \cite{PNR}, the size of user cache was fixed to be unity ($M_p=1$) as the users' private caches stored only the one-time pads used for transmissions. 
 Another network model similar to ours is the hierarchical content delivery network considered in \cite{KNM}. The hierarchical network in \cite{KNM} consists of two layers of caches arranged in a tree-like structure with the server as root node and the users are connected to the leaf caches. Each leaf cache serves only a single user. Each parent cache communicates only with its children caches. Hence, there is no direct communication between the server and the users. Each user is able to retrieve its demanded file using the contents of the leaf cache it is accessing and the messages received from the adjacent parent cache. For this hierarchical network, a coded caching scheme based on decentralized placement was designed in \cite{KNM} to minimize the number of transmissions made in various layers. In our model, the placement is centralized and the communication between the server and the users happens over a single-hop.
 
  %In this work, we consider the centralized placement and the communication between the server and the users happens over a single-hop.
 
  \begin{figure}[t]
 	\begin{center}
 		\captionsetup{justification=centering}
 		\includegraphics[width=0.5\textwidth]{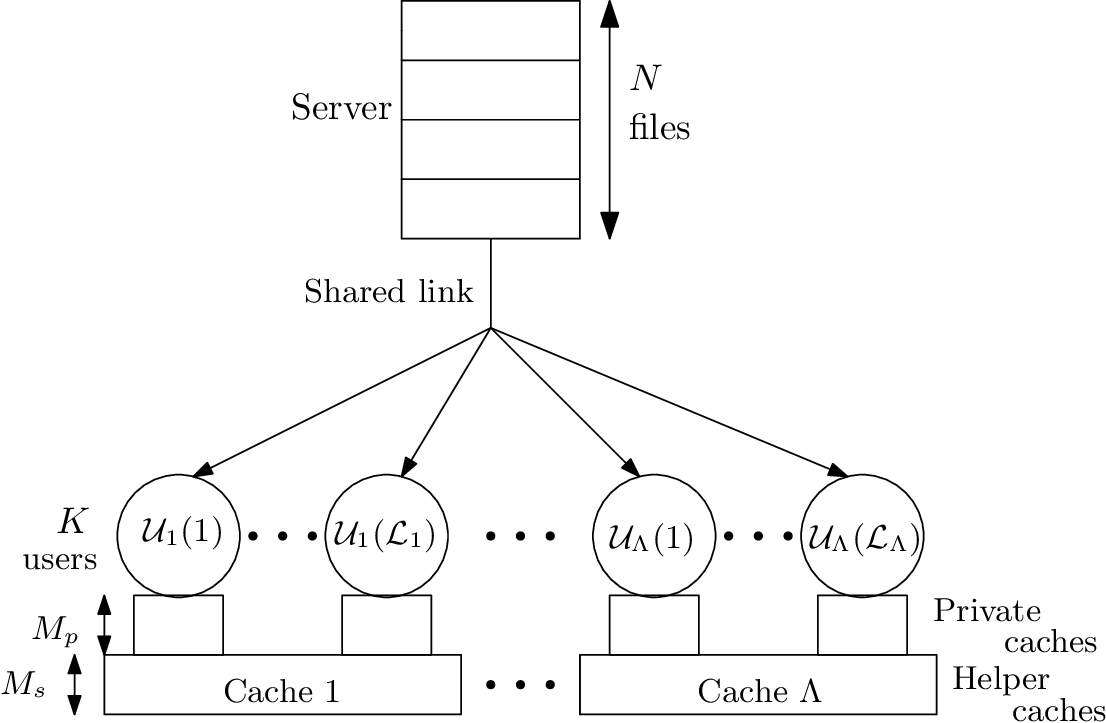}
 		\caption{Network model.}
 		\label{fig:setting}
 	\end{center}
 \end{figure}
 
 \subsection*{Our Contributions}
 We introduce a new network model where the helper caches and the users' dedicated caches coexist. The network model can be viewed as a generalization of the well-studied dedicated and shared cache networks. We study the network under two scenarios, one with prior knowledge of user-to-helper cache association and the other without any prior information about it. Our contributions are summarized below.
 \begin{itemize}
 		\item A centralized coded caching scheme is proposed for the case when the server has no information about the user-to-helper cache association during the placement phase (Theorem~\ref{thm1}). The scheme partitions each file into two parts of size $f$ and $(1-f)$ file units, respectively, where $f\in [0,1]$. Then, the optimal schemes, under uncoded placement, known for the dedicated \cite{MaN} and shared cache \cite{PUE} networks are  used to store and deliver each part of the library independently. In this way, maximum local caching gain is ensured for each user. Under the placement followed and for any value of $f$, the delivery policy is shown to be optimal using an index coding based converse. The rate achieved by the scheme is a function of $f$, and the parameter $f$ is chosen such that it minimizes the rate for a uniform user-to cache association (Section~\ref{sec:results}).	
 	\item Four achievable schemes are proposed for the case when the server is aware of the user-to-helper cache association during the placement phase.
 	\begin{enumerate}
 			\item  The first scheme is derived from the Maddah-Ali-Niesen scheme \cite{MaN}. Even though the applicability of the scheme is limited to a certain memory regime, the scheme is optimal (under uncoded placement) in the regime where it exists (Scheme $1$, Section~\ref{subsec:Scheme1}).
 		\item The second scheme follows from the scheme in Section~\ref{sec:results}, where the file divison is according to $f^{*}$ which is obtained by optimizing the rate with respect to the actual user-to-cache association (Section~\ref{subsec: profile_opt}). 
 		\item  The third scheme consists of novel placement and delivery techniques, where
 		each coded message broadcasted in the delivery phase is composed of the subfiles stored in both shared and private caches. 
 		This joint delivery scheme is designed without compromising on the local caching gain. The scheme is referred as Scheme $2$ later in the paper (Section~\ref{subsec:scheme2}). The scheme is also shown to be optimal (under no constraints on the cache placement) in certain memory regimes. The optimality is characterized by deriving a cut-set based information-theoretic lower bound on the optimal rate-memory tradeoff  (Section~\ref{sec:lb}). 
 		\item The fourth scheme is called as composite scheme,  in which the file library is optimally shared between Maddah-Ali Niesen scheme and our Scheme 2. The new scheme leverages the advantages provided by both the previous two schemes (Scheme in Section~\ref{subsec: profile_opt} and Scheme 2), and exhibits a better performance in all memory regimes compared to the scheme (Theorem~\ref{thm1}) in the user-to-cache association agnostic case (Section~\ref{subsec:MAN_scheme2}).
 	\end{enumerate} 
 \end{itemize}
 
\noindent  \textit{Notations:}
 For any integer $n$, $[n]$ denotes the set $\{1,2,\ldots,n\}$. For two positive integers $m,n$ such that $m <n$, $(m,n]$ denotes the set $\{m+1,m+2,\ldots,n\}$, and $[m,n]$ denotes the set $\{m,m+1,\ldots,n\}$. For a set $\{s_1,s_2,\ldots,s_n\}$, $s_{[1:m]}$ denotes the set $\{s_1,s_2,\ldots,s_m\}$, where $m <n$. The cardinality of a set $\mathcal{S}$ is denoted by $|\mathcal{S}|$. Binomial coefficients are denoted by $\binom{n}{k}$, where $\binom{n}{k}\triangleq \frac{n!}{k!(n-k)!}$ and $\binom{n}{k}=0$ for $n <k$. For a finite set $\mathcal{I}$, the lower convex envelope of the points $\{(i,f(i)): i \in \mathcal{I}\}$ is denoted by $Conv(f(i))$. 
 \par
 The remainder of the paper is organized as follows: Section~\ref{sec:setting} presents the network model and the problem setting. Section~\ref{sec:results} describes the coded caching scheme for the case when the server is not aware of the user-to-helper cache association prior to the placement. In Section~\ref{sec:setting2}, the other setting is considered, and the corresponding schemes are described. The numerical comparisons of the proposed schemes are given in Section~\ref{sec:analyses}. Section~\ref{sec:lb} presents a cut-set based lower bound for the network model, and lastly, in Section~\ref{sec:concl}, the results are summarized.

 \section{Problem Setting}
 \label{sec:setting}
 We consider a network model as illustrated in Fig.~\ref{fig:setting}. There is a server with access to a library of $N$ equal-length files $\mathcal{W}=\{W_1, W_2, \ldots, W_N\}$, and is connected to $K$ users through an error-free broadcast link. There are $\Lambda \leq K$ helper nodes that serve as caches, each of size equal to $M_s$ files. Each user $k \in [K]$ gets connected to one of the helper caches, and each user possesses a dedicated cache of size equal to $M_p$ files. The user-to-helper cache association can be arbitrary. We assume that $N \geq K$ and $M \leq N$, where $M \triangleq M_s+M_p$. We consider the two scenarios described in the following two subsections.
 \subsection{Server has no prior information about the association of users to helper caches}
 \label{subsec:settinga}
 When the server has no prior knowledge about the user-to-helper cache association, the system operates in three phases:

 	\textit{1) Placement phase}: In the placement phase, the server fills the helper caches and the user caches with the file contents adhering to the memory constraint. The contents can be kept in coded or uncoded form. During the placement, the server is unaware of the users' demands and the set of users accessing each helper cache. The above assumption is based on the fact that the content placement happens during off-peak times at which the users need not be necessarily connected to the helper caches. The contents stored at the helper cache $\lambda$ are denoted by $\mathcal{Z}_{\lambda}$, and the contents stored at the $k^{\textrm{th}}$ user's private cache are denoted by $Z_k$.
 		
 	\textit{2) User-to-helper cache association phase}: The content placement is followed by a phase where each user can connect to one of the $\Lambda$ helper caches. The set of users connected to the helper cache $\lambda \in [\Lambda]$ is denoted by $\mathcal{U}_{\lambda}$. The server gets to know each  $\mathcal{U}_{\lambda}$ from the respective helper caches.
 	Let $\mathcal{U}=\{\mathcal{U}_1,\mathcal{U}_2,\ldots,\mathcal{U}_{\Lambda}\}$ denote the user-to-helper cache association of the entire system. The set $\mathcal{U}$ forms a partition on the set of users $[K]$. The users connected to a helper cache can access the contents of the helper cache at zero cost. The number of users accessing each helper cache $\lambda$ is denoted by $\mathcal{L}_{\lambda}$, where $\mathcal{L}_{\lambda}=|\mathcal{U}_{\lambda}|$. Then, $\mathcal{L}=(\mathcal{L}_1,\mathcal{L}_2,\ldots,\mathcal{L}_{\Lambda})$ gives the overall association profile. Without loss of generality, it is assumed that $\mathcal{L}$ is arranged in the non-increasing order. 
 	
 	\textit{3) Delivery phase}: The delivery phase commences with the users revealing their demands to the server. Each user $k \in [K]$ demands one of the $N$ files from the server. We assume that the demands of the users are distinct. The demand of user $k \in [K]$ is denoted by $d_k$, and the demand vector is denoted as $\mathbf{d}=(d_1,d_2,\ldots,d_K)$. On receiving the demand vector $\mathbf{d}$, the server sends coded messages to the users over the shared link. Each user recovers its requested file using the available cache contents and the transmissions. 
 	
 \subsection{Server knows the user-to-helper cache association a priori}
 \label{subsec:settingb}
 When the server knows the user-to-helper cache association $\mathcal{U}$ before, the system operates in two phases: \textit{1) placement phase} and \textit{2) delivery phase}. The server fills the helper caches and the user caches by taking into account the user-to-helper cache association,  $\mathcal{U}$.
 
 \subsection*{Performance measure}
  In both cases, the overall length of the transmitted messages normalized with respect to unit file length is called rate and is denoted by $R(M_s,M_p)$. Our interest is in the worst-case rate, which corresponds to a case where each user demands a distinct file. The aim of any coded caching scheme is to satisfy all the user demands at a minimum rate. The optimal worst-case rate, denoted by $R^*(M_s,M_p)$, is defined as
\begin{equation*}
  R^{*}(M_s,M_p) \triangleq \inf\{R: R(M_s,M_p) \textrm{ is achievable}\}.
\end{equation*}
	Let $R^{*}_{\textrm{ded}}(M)$ and $R^{*}_{\textrm{shared}}(M)$ denote the optimal worst-case rates for dedicated and shared cache networks, respectively. The following lemma gives bounds on $R^{*}(M_s,M_p)$ in terms of $R^{*}_{\textrm{ded}}(M)$ and $R^{*}_{\textrm{shared}}(M)$.

\begin{lem}
	\label{lem1}
	For a network with $K$ users, each having a cache of size $M_p$ files and has access to one of the $\Lambda$ helper caches, each of size $M_s$ files, the optimal worst-case rate $R^{*}(M_s,M_p)$ satisfies the inequality
	\begin{equation}
	  R^{*}_{\textrm{ded}}(M) \leq R^{*}(M_s,M_p) \leq R^{*}_{\textrm{shared}}(M),
	  \label{eq:bound2}
	\end{equation}
	where $M \triangleq M_s+M_p$.
\end{lem}
\begin{IEEEproof}
	Consider a coded caching scheme for a shared cache network with $\Lambda$ caches, each of size $M$ files, that achieves the optimal rate-memory tradeoff $R^{*}_{\textrm{shared}}(M)$. At $M=M_s+M_p$, the above scheme achieves the optimal rate $R^{*}_{\textrm{shared}}(M)$. Consider a cache $s \in [\Lambda]$ in the shared cache network. The contents stored in the cache $s$ are denoted by $\mathcal{Z}_{s}$, where $|\mathcal{Z}_{s}|=M$ files. Then, $\mathcal{Z}_{s}$ can be written as the union of two disjoint sets as: $\mathcal{Z}_{s}=\mathcal{Z}_{\lambda}\cup {Z}_{p}$, where $|\mathcal{Z}_{\lambda}|=M_s$ files and $|{Z}_p|=M_p$ files. 
	
	Now, consider the network model in Fig.~\ref{fig:setting}. The placement followed in this setting is as follows: the content stored in the helper cache, $\lambda \in [\Lambda]$, is $\mathcal{Z}_{\lambda}$, and the content stored in the private cache of all the users in $\mathcal{U}_{\lambda}$ is $Z_{p}$. If we follow the delivery policy of the optimal shared cache scheme, we obtain $R(M_s,M_p)=R^{*}_{\textrm{shared}}(M)$. Thus, we can conclude that $R^{*}(M_s,M_p) \leq R^{*}_{\textrm{shared}}(M)$. 
	
	To  prove $ R^{*}_{\textrm{ded}}(M) \leq R^{*}(M_s,M_p)$, assume $Z^{*}$ and $D^{*}$ to be the placement and delivery policies that result in $R^{*}(M_s,M_p)$. The contents available to each user $k \in [K]$ from the placement $Z^{*}$ is $\mathcal{Z}_{\lambda_k} \cup Z_k$, where $\lambda_k$ is the helper cache accessed by the user $k$, $|\mathcal{Z}_{\lambda_k}|=M_s$ files, and $|Z_k|=M_p$ files. In a dedicated cache network with $K$ users, each having a cache of size $M=M_s+M_p$ files, it is possible to follow a placement such that the contents available to each user $k$ is exactly same as $\mathcal{Z}_{\lambda_k} \cup Z_k$. Then, following the delivery scheme $D^{*}$ results in  $R_{\textrm{ded}}(M)=R^{*}(M_s,M_p)$. Thus, we obtain $R^{*}_{\textrm{ded}}(M) \leq R^{*}(M_s,M_p)$. This completes  the proof of Lemma~\ref{lem1}.
\end{IEEEproof}
   
	Under uncoded placement, the optimal worst-case rate for the network model in Fig.~\ref{fig:setting} is denoted by $R^{*}_{\textrm{uncoded}}(M_s,M_p)$. Under the constraint of uncoded placement, the MaN scheme \cite{MaN}, and the shared cache scheme in \cite{PUE} achieve the optimal rate-memory curves for dedicated \cite{WTP} and shared cache \cite{PUE} networks, respectively. We denote the rate-memory tradeoff of the MaN scheme by $R^{*}_{\textrm{MaN}}(M)$, and the rate-memory tradeoff of the shared cache scheme in \cite{PUE} by $R^{*}_{\textrm{PUE}}(M)$. Then, the following corollary gives a bound on $R^{*}_{\textrm{uncoded}}(M_s,M_p)$ using $R^{*}_{\textrm{MaN}}(M)$ and $R^{*}_{\textrm{PUE}}(M)$.
	\begin{corollary}
	Under uncoded placement, we have 
	\begin{align*}
	  Conv \left( \frac{K-t_p}{t_p+1} \right)  & \leq R^{*}_{\textrm{uncoded}}(M_s,M_p) \leq\\ &  Conv \left( \frac{\sum_{n=1}^{\Lambda-t_s}\mathcal{L}_n\binom{\Lambda-n}{t_s}}{\binom{\Lambda}{t_s}} \right),
	\end{align*}
	where $t_p\triangleq\frac{KM}{N} \in [0,K]$, $t_s \triangleq \frac{\Lambda M}{N} \in [0,\Lambda]$, and $\mathcal{L}$ denotes the association profile.
	\end{corollary}
    \begin{IEEEproof}
      The proof directly follows from \eqref{eq:bound2}, since $R^{*}_{\textrm{MaN}}(M)=Conv\left( \frac{K-t_p}{t_p+1} \right)$, and $R^{*}_{\textrm{PUE}}(M)=Conv \left( \frac{\sum_{n=1}^{\Lambda-t_s}\mathcal{L}_n\binom{\Lambda-n}{t_s}}{\binom{\Lambda}{t_s}} \right)$, where $t_p = KM/N$ and $t_s = \Lambda M/N$.
    \end{IEEEproof}

\section{Networks with User-to-Helper Cache Association Unknown at the Placement Phase}
\label{sec:results}
In this section, we present a coded caching scheme for the setting discussed in Section~\ref{subsec:settinga}.
\begin{thm}
For a broadcast channel with $K$ users, each having a private cache of normalized size $M_p/N$, assisted by $\Lambda$ helper caches, each of normalized size ${M_s}/{N}$, the worst-case rate
\begin{equation}
\begin{aligned}
R(M_s,M_p) =  f R_s(t_s) +
(1-f)R_p(t_p)
\label{rate}
\end{aligned}
\end{equation}
 is achievable for an association profile $\mathcal{L}=(\mathcal{L}_1,\mathcal{L}_2,\ldots,\mathcal{L}_{\Lambda})$ and for every $f$ such that $  \frac{M_s}{N}\leq f \leq 1-\frac{M_p}{N}$, where $t_s \triangleq \frac{\Lambda M_s}{fN}$ and $t_p \triangleq \frac{KM_p}{(1-f)N}$. The terms $R_s(t_s)$ and $R_p(t_p)$ are defined as follows:
\begin{subequations}
\begin{align}
  R_s(t_s)& = (1-(t_s-\floor{t_s}))R_s(\floor{t_s})+(t_s-\floor{t_s})R_s(\ceil{t_s}),\\
  R_p(t_p) &= (1-(t_p-\floor{t_p}))R_p(\floor{t_p})+(t_p-\floor{t_p})R_p(\ceil{t_p}),
\end{align}
\label{eq:rs_rp}
\end{subequations}
where $R_s(u) = \frac{\sum_{n=1}^{\Lambda-u}\mathcal{L}_n\binom{\Lambda-n}{u}}{\binom{\Lambda}{u}}$ for $u \in [0, \Lambda]$, 
 and $R_p(v)=  \frac{K-v}{v+1}$ for $v \in [0,K]$.  
\label{thm1}
\end{thm}

\begin{IEEEproof}
%Consider a network as shown in Fig.~\ref{fig:setting}. The centralized scheme that achieves the performance in \eqref{rate} is presented in the sequel.
 Consider a network as shown in Fig.~\ref{fig:setting}. The three phases involved in the scheme that achieves the performance in \eqref{rate} are described below.
\subsubsection{Placement phase}
\label{subsec:placem}
There are two sets of caches in the system: helper caches and users' private caches. Each file  $W_n \in \mathcal{W}$ is split into two parts: $W_n = \{W_n^{(s)},W_n^{(p)}\}$ such that $|W_n^{(s)}|=f$ file units and $|W_n^{(p)}|=(1-f)$ file units, where  $M_s/N \leq f \leq 1-M_p/N$. The helper caches are filled with contents from the set $\{W_n^{(s)}, \forall n \in [N]\}$, and the user caches are filled with contents from the set $\{W_n^{(p)}, \forall n \in [N]\}$. We define, $t_s\triangleq \frac{\Lambda M_s}{fN}$ and $t_p \triangleq \frac{KM_p}{(1-f)N} $.

We first describe the helper cache placement. If $t_s \in [0,\Lambda]$, each subfile in the set $\{W_n^{(s)}, \forall n \in [N]\}$ is divided into $\binom{\Lambda}{t_s}$ mini-subfiles, and each mini-subfile is indexed by a set $\tau$, where $\tau \subseteq [\Lambda]$ and $|\tau|=t_s$. Then, the helper cache placement is defined as
\begin{equation}
\mathcal{Z}_{\lambda} = \{W_{n,\tau}^{(s)} \textrm{\hspace{0.05cm}}, \forall n \in [N]: \tau \ni \lambda\}, \textrm{\hspace{0.2cm}} \forall \lambda \in [\Lambda].
\label{eq:scheme1placem}
\end{equation}
\noindent Each helper cache $\lambda$ stores $N\binom{\Lambda-1}{t_s-1}$ number of mini-subfiles, each of  size $f/\binom{\Lambda}{t_s}$ file units. When $t_s \not\in [0,\Lambda]$, the memory-sharing technique needs to be employed. Thus, the  helper cache placement satisfies the memory constraint. 

To fill the user caches, each subfile in the set $\{W_n^{(p)}, \forall n \in [N]\}$ is split into $\binom{K}{t_p}$ mini-subfiles, if $t_p \in [0,K]$. Otherwise, the memory-sharing technique is employed. The mini-subfiles are indexed by sets $\rho \subseteq [K]$ such that $|\rho|=t_p$. The file contents placed in the $k^{th}$ user's private cache are given by
$Z_k = \{W_{n,\rho}^{(p)} \textrm{\hspace{0.05cm}}, \forall n \in [N]: \rho \ni k \}\textrm{\hspace{0.05cm}}.$ The user cache placement also obeys the memory constraint. Consider a user $k \in [K]$, there are $N\binom{K-1}{t_p-1}$ number of mini-subfiles with $k \in \rho$. Each mini-subfile is of size $(1-f)/\binom{K}{t_p}$ file units. Thus, we obtain $N\binom{K-1}{t_p-1}(1-f)/\binom{K}{t_p}=M_p$.

\subsubsection{User-to-helper cache association phase}
The placement is done without knowing the identity of the users getting connected to each helper cache. Each user accesses one of the helper caches. Let the user-to-helper cache association be $\mathcal{U}$ with a profile $\mathcal{L}=(\mathcal{L}_1,\mathcal{L}_2,\ldots,\mathcal{L}_{\Lambda})$.

\subsubsection{Delivery phase}
In the delivery phase, each user demands one of the $N$ files. We consider a distinct demand scenario.  Let the demand vector be $\mathbf{d}=(d_1,d_2,\ldots,d_K)$. Since each demanded file consists of two parts, $W_n^{(s)}$ and $W_n^{(p)}$, the delivery also consists of two sets of transmissions. We first consider the set $\{W_n^{(s)}, \forall n \in \mathbf{d}\}$.  The subfiles belonging to this set are transmitted using the shared cache delivery scheme in \cite{PUE}. For every $j \in [\mathcal{L}_1]$, the server forms sets $\mathcal{S} \subseteq [\Lambda]$, $|\mathcal{S}|=t_s+1$, and makes the transmission
$ X_{\mathcal{S},j}= \underset{{\lambda \in \mathcal{S}:\mathcal{L}_{\lambda}\geq j}}{\bigoplus}W^{(s)}_{d_{\mathcal{U}_{\lambda}(j)},\mathcal{S}\backslash\lambda}\textrm{\hspace{0.05cm}}.
$
 
Next, consider the set $\{W^{(p)}_n, \forall n\in \mathbf{d}\}$. The subfiles belonging to this set are delivered using the transmission scheme in \cite{MaN}. For each set $\mathcal{P} \subseteq [K]$ such that $|\mathcal{P}|=t_p+1$, the server transmits the message
$ X_{\mathcal{P}}= \underset{k \in \mathcal{P}}{\bigoplus}W^{(p)}_{d_k,\mathcal{P}\backslash k}\textrm{\hspace{0.05cm}}.
$
\par
\textit{Decoding:} We first look at the decoding of the subfiles of the set $\{W^{(s)}_{n}, \forall n \in \mathbf{d}\}$. Consider a transmission ${X}_{\mathcal{S},j}$, and a user $\mathcal{U}_{\lambda}(j)$ benefiting from $X_{\mathcal{S},j}$. Then, the message received at user $\mathcal{U}_{\lambda}(j)$ is of the form:
%\begin{equation}
$ W^{(s)}_{d_{\mathcal{U}_{\lambda}(j)},\mathcal{S}\backslash \lambda} \underbrace{\underset{\substack{\lambda^{\prime} \in \mathcal{S}: \mathcal{L}_{\lambda^{\prime}} \geq j\\ \lambda^{\prime}\neq \lambda}}{\bigoplus} W^{(s)}_{d_{\mathcal{U}_{\lambda^{\prime}}(j)},\mathcal{S}\backslash \lambda^{\prime} }}_\textrm{known to user $\mathcal{U}_{\lambda}(j)$}.$
%\label{eq:decsc}
%\end{equation}
\noindent Likewise, user $\mathcal{U}_{\lambda}(j)$ can decode the rest of the mini-subfiles of $W^{(s)}_{d_{\mathcal{U}_{\lambda}(j)}}$. The decoding of the subfiles associated with the set $\{W^{(p)}_{n},\forall n \in \mathbf{d}\}$ is also done in a similar way. The users' demands get served once the subfiles  $W^{(s)}_n$ and $W^{(p)}_n, \forall n \in \mathbf{d}$, are recovered.
\par
\textit{Performance measure:} The rate calculation can be decomposed into two parts as done previously for the placement and delivery phases. The rate associated with the transmission of the subfiles of the set $\{W^{(s)}_n, \forall n \in \mathbf{d}\}$ is calculated as follows: the number of transmissions associated with $(t_s+1)-$sized sets $\mathcal{S}$ containing $\min\{\lambda: \lambda \in \mathcal{S}\}=1$ is $\mathcal{L}_1$. Similarly, the number of transmissions associated with sets $\mathcal{S}$ containing $\min\{\lambda: \lambda \in \mathcal{S}\}=2$ is $\mathcal{L}_2$. By proceeding further in the same manner up to $\min\{\lambda: \lambda \in \mathcal{S}\}=\Lambda-t_s$, the total number of transmissions is obtained as $\sum_{n=1}^{\Lambda-t_s}\mathcal{L}_n\binom{\Lambda-n}{t_s}$, and each transmission is of size $f/\binom{\Lambda}{t_s}$ file units. Thus, the rate obtained for this case is ${ \sum_{n=1}^{\Lambda-t_s}\mathcal{L}_n\binom{\Lambda-n}{t_s}f}/{\binom{\Lambda}{t_s}}$.

The rate associated with the transmission of the subfiles $\{W^{(p)}_n,\forall n \in \mathbf{d}\}$ is calculated as follows:
there are $\binom{K}{t_p+1}$ number of transmissions, each of size $(1-f)/\binom{K}{t_p}$ file units. Therefore, the rate associated with the transmissions of type $X_{\mathcal{P}}$, $ \mathcal{P}\subseteq [K]$, $|\mathcal{P}|=t_p+1$, is ${\binom{K}{t_p+1}(1-f)}/{\binom{K}{t_p}}$.
Thus, the rate $R(M_s,M_p)$ is obtained as
\begin{align}
 R(M_s,M_p)&=\frac{ \displaystyle \sum_{n=1}^{\Lambda-t_s}\mathcal{L}_n \binom{\Lambda-n}{t_s}f}{\binom{\Lambda}{t_s}}+\frac{(K-t_p)(1-f)}{t_p+1},
 \label{eq:rate_sum}
\end{align}
when $t_s \in [0,\Lambda]$ and $t_p \in [0,K]$. The expressions in \eqref{eq:rs_rp} are obtained by considering the memory-sharing instances. While performing memory-sharing in shared caches, $1-(t_s-\floor{t_s})$ is the fraction of memory and the fraction of $W_n^{(s)}, \forall n \in [N]$, that are stored and delivered according to $\floor{t_s}$. The remaining $t_s-\floor{t_s}$ fractions of the helper cache memory and $W_n^{(s)}$  are filled and sent according to $\ceil{t_s}$. Similarly, $1-(t_p-\floor{t_p})$ and $t_p-\floor{t_p}$ are the fractions of dedicated cache memory allocated corresponding to $\floor{t_p}$ and $\ceil{t_p}$, respectively. Thus, $R_s(t_s)$ and $R_p(t_p)$ are obtained as follows:
\begin{equation*}
	\begin{aligned}
	R_s(t_s)& = (1-(t_s-\floor{t_s}))R_s(\floor{t_s})+(t_s-\floor{t_s})R_s(\ceil{t_s}),\\
	R_p(t_p) &= (1-(t_p-\floor{t_p}))R_p(\floor{t_p})+(t_p-\floor{t_p})R_p(\ceil{t_p}).
	\end{aligned}
\end{equation*}

Since the user-to-cache association is not known at the placement, the parameter $f$ is chosen such that it minimizes the rate $R(M_s,M_p)$ in \eqref{rate} for a uniform association profile. That is,
 \begin{align}
   f^{*} & =\underset{\frac{M_s}{N}\leq f \leq 1-\frac{M_p}{N}}{\arg\min}\textrm{\hspace{0.2cm}}fR_s(t_s)+(1-f)R_p(t_p), \label{eq:f_value1}
   \end{align}
   Thus, we obtain the required value as:
   \begin{small}
   \begin{align}
      f^\ast & = \frac{M_s}{N}\max\left(1,\frac{(1+\frac{KM_p}{N})\sqrt{\Lambda(\Lambda+1)} - \frac{\Lambda M_p}{N}\sqrt{K(K+1)}}{\sqrt{K(K+1)}\frac{M_p}{N}+ \sqrt{\Lambda(\Lambda+1)}\frac{M_s}{N}}\right).
   \label{eq:f_value2}
\end{align}
\end{small}
 The transition from \eqref{eq:f_value1} to \eqref{eq:f_value2} is given in Appendix~\ref{sec:appendix1}. In the case of uniform profiles, for $u \in [0,\Lambda]$, we have $R_s(u)=\frac{K(\Lambda-u)}{\Lambda(u+1)}$. Note that the value of $f$ that minimizes the rate  for a uniform profile will not result in the least $R(M_s,M_p)$ for an arbitrary profile. However, the rationale behind considering the uniform profile is that it results in the least rate among all the profiles for a given shared cache system. Whereas, the profile $\mathcal{L}=(K,0,\ldots,0)$ is not considered as it does not support any multicasting opportunity, and the obvious choice of $f^{*}$ is $M_s/N$ in that case.
 %This completes the proof of Theorem~\ref{thm1}.
\end{IEEEproof}
\begin{rem}
	By approximating $\sqrt{K(K+1)}\approx K$ and $\sqrt{\Lambda(\Lambda+1)}\approx \Lambda$ in \eqref{eq:f_value2}, we get $f^\ast = \max\left(\frac{M_s}{N},\frac{\Lambda M_s}{\Lambda M_s+KM_p}\right)$. % In order to minimize the rate, we employ the optimal shared cache scheme \cite{PUE} on $f^\ast$ fraction of the server file library and the Maddah-Ali-Niesen scheme \cite{MaN} (optimal scheme for dedicated coded caching) on the rest $(1-f^\ast)$ fraction of the library.
	 Note that the fraction $f^\ast = \frac{\Lambda M_s}{\Lambda M_s+KM_p}$ is the ratio of the total helper cache memory in the system, $\Lambda M_s$ to the total system memory $\Lambda M_s+KM_p$. Further, if $f^\ast = {M_s}/{N}$, the same content will be stored across all the helper caches, and there will be no transmission corresponding to the optimal shared cache scheme ($f^\ast = M_s/N \implies t_s = \Lambda$).  
\end{rem}

\begin{rem}
	When $M_s=0$, we obtain $f^{*}=0$, and the proposed scheme reduces to the optimal MaN scheme in \cite{MaN}. Likewise, when $M_p=0$, we get $f^{*}=1$, and the proposed scheme recovers the optimal shared cache scheme in \cite{PUE}.
\end{rem}

\subsection{Converse}
 In this subsection, we show that for any $\frac{M_s}{N} \leq f \leq 1-\frac{M_p}{N}$, our delivery scheme is optimal under the placement policy described in the proof of Theorem~\ref{thm1}. The converse is derived using index coding techniques. Before proving the optimality, we take a detour into the index coding problem, and see its relation to the coded caching problem.

\subsubsection{Preliminaries on index coding}
The index coding problem considered in \cite{ABKSW} consists of a sender who wishes to communicate $L$ messages $x_j$, $j \in [L]$, to a set of $L$ receivers over a shared noiseless link. Each message $x_j$ is of size $|x_j|$ bits. Each receiver $j \in [L]$ wants $x_j$, and knows a subset of messages as side-information. Let $\mathcal{K}_j$ denote the indices of the messages known to the $j^{th}$ receiver. Then, $\mathcal{K}_j \subseteq [1:L] \backslash \{j\}$. To satisfy the receivers' demands, the sender transmits a message $X$ of length $l$ bits. Based on the side-information $\mathcal{K}_j$ and the received message $X$, each receiver $j \in [L]$ retrieves the message $x_j$. This index coding problem can be represented using a directed graph $\mathcal{G}=(\mathcal{V},\mathcal{E})$, where $\mathcal{V}$ is the set of vertices and $\mathcal{E}$ is the set of directed edges of the graph. Each vertex in $\mathcal{G}$ represents a receiver $j$ and its demanded message $x_j$. A directed edge from vertex $i$ to vertex $j$ represents that the receiver $j$ knows the message $x_i$.

The following lemma gives a bound on the transmission length $l$ using $\mathcal{G}$ and is drawn from [Corollary 1, \cite{ABKSW}].

\begin{lem}[Cut-set type converse \cite{ABKSW}]
For an index coding problem represented by the directed graph $\mathcal{G}=(\mathcal{V},\mathcal{E})$, if each receiver $j \in [L]$ retrieves its wanted message from the transmission $X$ using its side-information, then the following inequality holds
\begin{equation}
\sum_{j \in \mathcal{J}} \frac{|x_j|}{l} \leq 1
\label{eq:index}
\end{equation}	
for every $\mathcal{J} \subseteq [1:L]$ such that the subgraph of $\mathcal{G}$ over $\mathcal{J}$ does not contain a directed cycle.
\label{conv:lem}
\end{lem}
From \eqref{eq:index}, we get 
\begin{equation}
l \geq \sum_{j \in \mathcal{J}} {|x_j|} .
\label{eq:ratebound}
\end{equation}
By finding a maximum acylic subgraph of $\mathcal{G}$, we get a tighter lower bound on the transmission length, $l$.

\subsubsection{Relation between index coding and coded caching problems}
For a fixed placement and a demand, the delivery phase of a coded caching problem can be viewed as an index coding problem. Let $\mathbf{Z}$ denote the placement policy followed in the scheme. Then, the index coding problem induced by $\mathbf{Z}$, an association profile $\mathcal{L}$, and a distinct demand vector $\mathbf{d}$ is denoted by $\mathcal{I}(\mathbf{Z},\mathcal{L},\mathbf{d})$. The index coding problem $\mathcal{I}(\mathbf{Z}, \mathcal{L},\mathbf{d})$ has $K\big(\binom{\Lambda-1}{t_s}+\binom{K-1}{t_p}\big)$ messages, where each message corresponds to a subfile of the demanded files which is not available to the users from their caches. Thus, the messages in $\mathcal{I}(\mathbf{Z},\mathcal{L,\mathbf{d}})$ are constituted by the subfiles $W^{(s)}_{d_i,\tau}$ and $W^{(p)}_{d_i,\rho}$, where $i \in [K]$, $\rho \not\ni i$, and $\tau \not\ni c(i)$ which is the helper cache accessed by user $i$. In an index coding problem, if  a receiver demands multiple messages, it can be equivalently seen as that many receivers with the same side-information as before, and each demanding a single message. Therefore, the number of receivers in $\mathcal{I}(\mathbf{Z},\mathcal{L},\mathbf{d})$ is $K\big(\binom{\Lambda-1}{t_s}+\binom{K-1}{t_p}\big)$. Thus, the graph $\mathcal{G}$ representing $\mathcal{I}(\mathbf{Z},\mathcal{L},\mathbf{d})$ contains $K\big(\binom{\Lambda-1}{t_s}+\binom{K-1}{t_p}\big)$ vertices. Each vertex represents a subfile and the user demanding it. The edges in $\mathcal{G}$ are determined by the cache contents known to each user. Let $W^{(s)}_{d_{i_1}, \tau_1}$ and $W^{(s)}_{d_{i_2}, \tau_2}$ be two vertices in $\mathcal{G}$ such that $i_1,i_2 \in [K]$ and $i_1 \neq i_2$. Let $c(i_1)$ and $c(i_2)$ denote the helper caches to which the users $i_1$ and $i_2$ are connected, respectively. Then, there exists a directed edge from $W^{(s)}_{d_{i_1}, \tau_1}$ to $W^{(s)}_{d_{i_2}, \tau_2}$ if $c(i_2) \in \tau_1$, and vice versa if $c(i_1) \in \tau_2$. There does not exist a directed edge between $W^{(s)}_{d_{i_1}, \tau_1}$ and $W^{(s)}_{d_{i_2}, \tau_2}$ if $c(i_1)=c(i_2)$. Similarly, assume $W^{(p)}_{d_{i_1},\rho_1}$ and $W^{(p)}_{d_{i_2},\rho_2}$ to be two vertices in $\mathcal{G}$. There exists a directed edge from $W^{(p)}_{d_{i_1},\rho_1}$ to $W^{(p)}_{d_{i_2},\rho_2}$ if $i_2 \in \rho_1$, and vice versa if $i_1 \in \rho_2$. Next, consider the vertices $W^{(s)}_{d_{i_1}, \tau_1}$ and $W^{(p)}_{d_{i_2},\rho_2}$. There exists a directed edge from $W^{(s)}_{d_{i_1}, \tau_1}$ to $W^{(p)}_{d_{i_2},\rho_2}$ if $c(i_2) \in \tau_1$, and vice versa if $i_1 \in \rho_2$.

According to Lemma~\ref{conv:lem}, a subgraph of $\mathcal{G}$ that does not contain any cycle needs to be constructed to obtain a lower bound on the achievable rate. For the sake of simplicity, assume each file to be of unit length. 
 Consider a subgraph $\mathcal{J}$ of $\mathcal{G}$ formed by the subfiles 
\begin{align}
\bigcup_{i \in [K]}\big\{W^{(s)}_{d_i,\tau}, W^{(p)}_{d_i,\rho}:\textrm{\hspace{0.1cm}} & \tau \subseteq [\Lambda] \backslash [c(i)], |\tau|=t_s, \\ &  \rho \subseteq [K] \backslash [i], |\rho|=t_p \big\}.
\label{eq:acyclic}
\end{align}
We need to show that $\mathcal{J}$ is an acyclic subgraph of $\mathcal{G}$. Recall that in the coded caching system, the helper caches are arranged in the non-increasing order of the number of users served by them. Without loss of generality, we assume that the users are also ordered in the caching system. The subfiles in \eqref{eq:acyclic} can be written as $H_1 \cup H_2$, where
\begin{align*}
& H_1 \triangleq \cup_{i \in [K]}\big\{W^{(s)}_{d_i,\tau}: \tau \subseteq [\Lambda] \backslash [c(i)], |\tau|=t_s \big\} \textrm{ and }\\
& H_2 \triangleq \cup_{i \in [K]} \big\{W^{(p)}_{d_i,\rho}: \rho \subseteq [K] \backslash [i], |\rho|=t_p \big\}.
\end{align*}
Consider any two subfiles $W^{(s)}_{d_{i_1},\tau_1}$ and $W^{(s)}_{d_{i_2},\tau_2} \in H_1$. If there exists an edge from $W^{(s)}_{d_{i_1},\tau_1}$ to $W^{(s)}_{d_{i_2},\tau_2}$, then $c(i_2) \in \tau_1$ which implies $c(i_2) > c(i_1)$. Therefore, it is not possible to have an edge going from $W^{(s)}_{d_{i_2},\tau_2}$ to  $W^{(s)}_{d_{i_1},\tau_1}$. Next, consider any two subfiles $W^{(p)}_{d_{i_1},\rho_1}$ and $W^{(p)}_{d_{i_2},\rho_2} \in H_2$. If there exists an edge from $W^{(p)}_{d_{i_1},\rho_1}$ to $W^{(p)}_{d_{i_2},\rho_2}$, then $i_2 \in \rho_1$ which implies $i_2 >i_1$. Therefore, an edge from $W^{(p)}_{d_{i_2},\rho_2}$ to $W^{(p)}_{d_{i_1},\rho_1}$ does not exist. 

Next, consider the subfiles $W^{(s)}_{d_{i_1},\tau_1}, W^{(s)}_{d_{i_2},\tau_2} \in H_1$ and $W^{(p)}_{d_{i_3},\rho_3},W^{(p)}_{d_{i_4},\rho_4}  \in H_2$. Consider a scenario where there is an edge $e_1$ emanating from $W^{(s)}_{d_{i_1},\tau_1}$ to $W^{(s)}_{d_{i_2},\tau_2}$, another edge $e_2$ from $W^{(s)}_{d_{i_2},\tau_2}$ to $W^{(p)}_{d_{i_3},\rho_3}$, and there is an edge $e_3$ from $W^{(p)}_{d_{i_3},\rho_3}$ to $W^{(p)}_{d_{i_4},\rho_4}$. We need to show that it is not possible to have another edge $e_4$ from $W^{(p)}_{d_{i_4},\rho_4}$ to $W^{(s)}_{d_{i_1},\tau_1}$. We prove this by contradiction. Assume that $e_4$ exists. Then, the above set of four subfiles form a cycle. Edge $e_1$ implies $c(i_2) \in \tau_1$ and $c(i_2) > c(i_1)$. Since the users are also ordered, we can say that the user index $i$ is greater than its helper cache index $c(i)$. Therefore, we obtain $i_2 > i_1$. Edge $e_2$ implies $c(i_3) \in \tau_2$ and $c(i_3) >c(i_2)$ which, in turn, gives $i_3 >i_2$. The relation $i_3 >i_2$ ensures that an edge from $W^{(p)}_{d_{i_3},\rho_3}$ to $W^{(s)}_{d_{i_2},\tau_2}$ does not exist in $\mathcal{J}$. Similarly, from edge $e_3$, we conclude that $i_4 >i_3$. Thus, we get $i_4>i_3>i_2>i_1$. Edge $e_4$ gives $i_1 > i_4$, which is not possible. Hence, edge $e_4$ does not exist. Along the same lines, we can show that edge $e_2$ does not exist in a case where edges $e_1$, $e_3$, and $e_4$ are present. Hence, we proved that the  subfiles in \eqref{eq:acyclic} form an acyclic subgraph $\mathcal{J}$ over the graph $\mathcal{G}$.

Using Lemma~\ref{conv:lem}, we can write $R(M_s,M_p)\geq |H_1|+|H_2|$. The number of subfiles in $H_1$ is obtained as: $\sum_{n=1}^{\Lambda}\mathcal{L}_n\binom{\Lambda-n}{t_s}$, and the number of subfiles in $H_2$ is obtained as: $\sum_{n=1}^{K}\binom{K-n}{t_p}=\binom{K}{t_p+1}$. Each subfile in $H_1$ is of size $f/\binom{\Lambda}{t_s}$ file units, and each subfile in $H_2$ is of size $(1-f)/\binom{K}{t_p}$ units, where $\frac{M_s}{N} \leq f \leq 1-\frac{M_p}{N}$. Thus, we get
\begin{align*}
R(M_s,M_p) \geq \frac{f\sum_{n=1}^{\Lambda}\mathcal{L}_n\binom{\Lambda-n}{t_s}}{\binom{\Lambda}{t_s}}  +\frac{(1-f)(K-t_p) }{t_p+1}. 
\end{align*}
From \eqref{eq:rate_sum}, we know that the rate 
\begin{align*}
R(M_s,M_p) = \frac{f\sum_{n=1}^{\Lambda}\mathcal{L}_n\binom{\Lambda-n}{t_s}}{\binom{\Lambda}{t_s}}  +\frac{(1-f)(K-t_p) }{t_p+1}
\end{align*}
is achievable. This proves the optimality of our delivery scheme under the placement policy described in the proof of Theorem~\ref{thm1}. The converse is derived assuming $t_s$ and $t_p$ are integers, even otherwise, the above converse holds in the memory-sharing cases. %\hfill \qedsymbol 

Now, we describe the reason behind following such a placement in the scheme. The cache contents known to user $k \in [K]$ are $Z_k$ and $\mathcal{Z}_{\lambda}$, where $\lambda$ is the helper cache accessed by user $k$. If $Z_k \cap \mathcal{Z}_{\lambda}=\phi$, then the amount of file contents available to each user from the placement phase is maximum. In addition, during the placement, the server is unaware of the user-to-helper cache association. The above two reasons made us divide each file into two parts, and each part being  exclusively used to fill one type of caches. From the previous results \cite{YMA}, it is known that the placement in \cite{MaN} is optimal under uncoded prefetching schemes. Therefore, we followed the optimal prefetching scheme in \cite{MaN} for the helper caches and the private caches using $\{W^{(s)}_n, \forall n \in [N]\}$ and $\{W^{(p)}_n, \forall n \in [N]\}$, respectively. Although the placement policies followed are individually optimal, the overall optimality of the above placement for this network model is not known.

We now illustrate the scheme and its optimality using an example.
\begin{example}
$N=4$, $K=4$, $\Lambda =2$, $M_s = 2$, $M_p=0.5$, and $\mathcal{L}=(3,1)$
\label{ex:exmp1}
\end{example}
Consider a scenario where there is a server with $N=4$ files, $\mathcal{W}=\{W_{1},W_{2}, W_{3},W_{4}\}$, each of unit size. The server is connected to $\Lambda=2$ helper caches and to $K=4$ users. Each helper cache is of size $M_s = 2$ units, and the private cache's size is half a unit.

Each file $W_n, n\in [4]$, is split into two subfiles: $W_n=\{W^{(s)}_n, W^{(p)}_n\}$. Since $f^{*}$ is obtained as $3/4$, the  size of each subfile is as follows: $|W^{(s)}_n|=3/4$ units and $|W^{(p)}_n|=1/4$ units. Therefore, we get $t_s = 4/3$ and $t_p =2$.

Since $t_s \not\in [0,2]$, the memory-sharing technique needs to be employed. Corresponding to $\floor{t_s} = 1$ and $\ceil{t_s}=2$, we get $M_{s_1}=3/2$ and $M_{s_2}=3$, respectively. Then, $M_s =\alpha M_{s_1} + (1-\alpha)M_{s_2} $, where $\alpha = 2/3$. Therefore, each subfile $W_{n}^{(s)}$ is first split into two parts: $W_n^{(s),(\alpha)}$ and $W_n^{(s),(1-\alpha)}$, each of size $2f^{*}/3$ and $f^{*}/3$ file units. The mini-subfiles of the sets $\{W_n^{(s),\alpha}, \forall n \in [N]\}$ and $\{W_n^{(s),(1-\alpha)}, \forall n \in [N]\}$  are placed in the helper caches according to \eqref{eq:scheme1placem}, using $\floor{t_s}$ and $\ceil{t_s}$ respectively. Thus, the contents stored in each helper cache are:
\begin{equation}
\begin{aligned}
\mathcal{Z}_1 &= \{W_{n,1}^{(s),(\alpha)},W_{n,12}^{(s),(1-\alpha)}, \forall n \in [4]\},\hspace{0.3cm}\\
\mathcal{Z}_2 &= \{W_{n,2}^{(s),(\alpha)},W_{n,12}^{(s),(1-\alpha)}, \forall n \in [4]\}.
\label{eq:cache}
\end{aligned} 
\end{equation}
The  user caches are filled using the set of subfiles $\{W_n^{(p)}, \forall n \in [4]\}$. Each subfile in $\{W_n^{(p)}, \forall n \in [4]\}$ is further divided into $6$ mini-subfiles. The users' cache contents are:
\begin{equation}
\begin{aligned}
Z_1 &=\big\{W^{(p)}_{n,12}, W^{(p)}_{n,13}, W^{(p)}_{n,14}\textrm{\hspace{0.05cm}},\forall n \in [4]\big\},\\
Z_2 & = \big\{W^{(p)}_{n,12}, W^{(p)}_{n,23}, W^{(p)}_{n,24}\textrm{\hspace{0.05cm}},\forall n \in [4]\big\},\\
Z_3 &= \big\{W^{(p)}_{n,13}, W^{(p)}_{n,23}, W^{(p)}_{n,34}\textrm{\hspace{0.05cm}},\forall n \in [4]\big\}, \\
Z_4 &= \big\{W^{(p)}_{n,14}, W^{(p)}_{n,24}, W^{(p)}_{n,34}\textrm{\hspace{0.05cm}}, \forall n \in [4]\big\}.
\label{eq:userplace}
\end{aligned}
\end{equation}
Note that the placement in the helper caches and the user caches satisfies the memory constraint. Assume that the user-to-helper cache association is $\mathcal{U}=\{\{1,2,3\},\{4\}\}$ with an association profile $\mathcal{L}=(3,1)$. For a demand vector $\mathbf{d}=(1,2,3,4)$, the transmissions are listed below:
\begin{align*}
  X_{\{12\},1}^{(\alpha)} & = W^{(s),(\alpha)}_{1,2} \oplus W^{(s),(\alpha)}_{4,1},\textrm{\hspace{0.3cm}}
   X_{\{12\},2}^{(\alpha)}  = W^{(s),(\alpha)}_{2,2}, \textrm{\hspace{0.3cm}}\\ X_{\{12\},3}^{(\alpha)} & = W^{(s),(\alpha)}_{3,2},
  X_{\{123\}} = W_{1,23}^{(p)} \oplus W_{2,13}^{(p)} \oplus W_{3,12}^{(p)},\textrm{\hspace{0.3cm}} \\ X_{\{124\}} & = W_{1,24}^{(p)} \oplus W_{2,14}^{(p)} \oplus W_{4,12}^{(p)},\\
  X_{\{134\}}  &= W_{1,34}^{(p)} \oplus W_{3,14}^{(p)} \oplus W_{4,13}^{(p)},\textrm{\hspace{0.3cm}} \\
  X_{\{234\}}& = W_{2,34}^{(p)} \oplus W_{3,24}^{(p)} \oplus W_{4,23}^{(p)}.
\end{align*}
The decoding is straightforward.  Each user can recover its desired subfiles from the above transmissions using its side-information. Thus, the rate achieved is $R(2,1/2)= 3/4 + 1/6=11/12$.
\par

\textit{Optimality}: Under the placement in \eqref{eq:cache} and \eqref{eq:userplace}, we show that the obtained rate is optimal. Let the placement performed in \eqref{eq:cache} and \eqref{eq:userplace} be jointly represented as $\mathbf{Z}$. Then, the index coding problem $\mathcal{I}(\mathbf{Z},\mathcal{L},\mathbf{d})$ induced by $\mathbf{Z},\mathcal{L}=(3,1),$ and $\mathbf{d}=(1,2,3,4)$ consists of $16$ messages and $16$ receivers (in the coded caching problem, each of the $4$ user wants $4$ more subfiles of its demanded file). This index coding problem can be represented using a directed graph $\mathcal{G}$ containing $16$ vertices, each simultaneously representing a user and a mini-subfile. The mini-subfiles in $\mathcal{G}$ are:  $W^{(s),(\alpha)}_{1,2},W^{(s),(\alpha)}_{2,2},W^{(s),(\alpha)}_{3,2},W^{(s),(\alpha)}_{4,1}$, $ W^{(p)}_{1,23},W^{(p)}_{1,24},W^{(p)}_{1,34},W^{(p)}_{2,13},W^{(p)}_{2,14},W^{(p)}_{2,34},W^{(p)}_{3,12},W^{(p)}_{3,14},$ $W^{(p)}_{3,24},W^{(p)}_{4,12},W^{(p)}_{4,13},W^{(p)}_{4,23}$. Then, the following set of mini-subfiles
 $\{W^{(s),(\alpha)}_{1,2},W^{(s),(\alpha)}_{2,2},W^{(s),(\alpha)}_{3,2}, W^{(p)}_{1,23}$, $W^{(p)}_{1,24},W^{(p)}_{1,34},W^{(p)}_{2,34}\}$ form an acyclic subgraph $\mathcal{J}$ over $\mathcal{G}$. The acyclic subgraph $\mathcal{J}$ is shown in Fig.~\ref{fig:graph}. Then, we get $R(1,1) \geq 3/4 + 1/6 =11/12$. Thus, the optimality of the delivery scheme is proved under the placement followed.

 \begin{figure}[h]
	\begin{center}
		\captionsetup{justification=centering}
		\includegraphics[width=0.85\columnwidth]{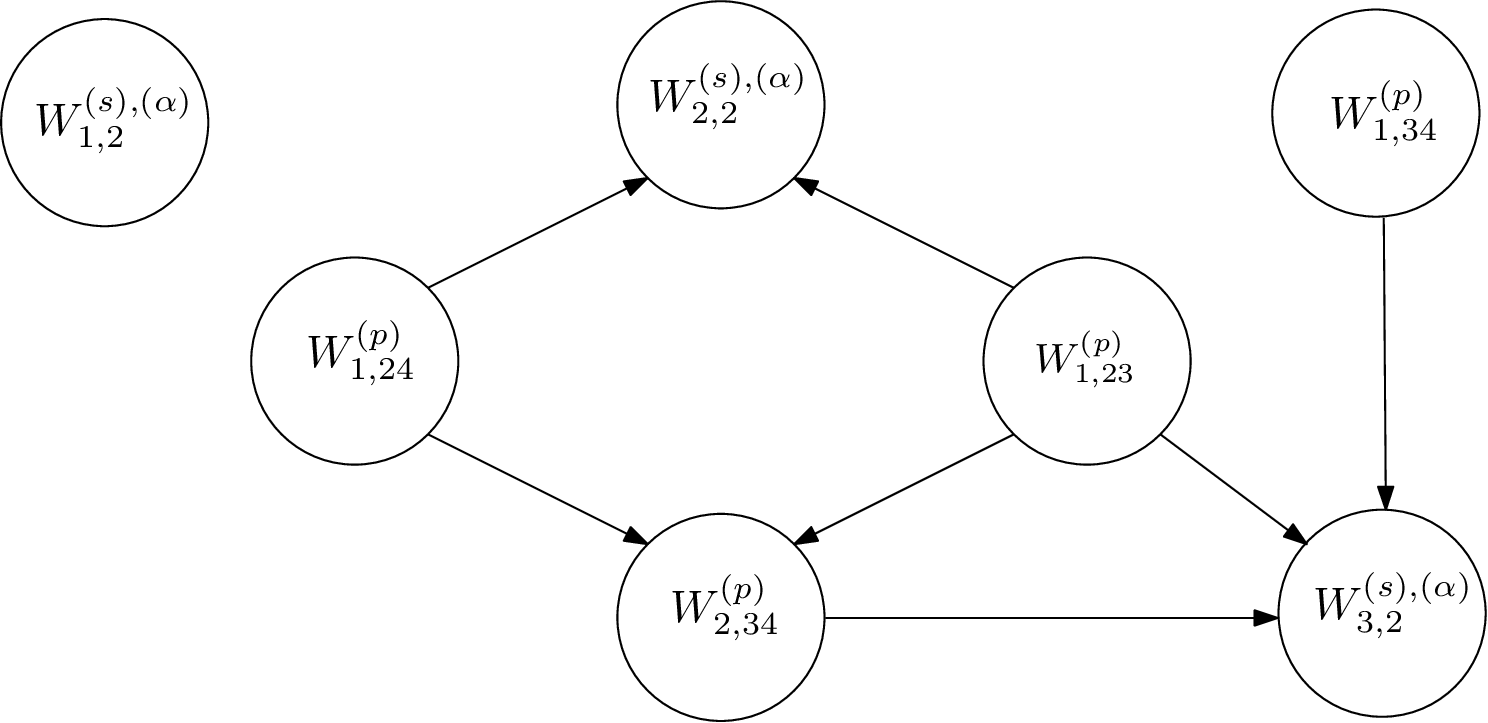}
		\caption{An acyclic subgraph $\mathcal{J}$.}
		\label{fig:graph}
	\end{center}
\end{figure}

\section{Networks with User-to-Helper Cache Association Known}
\label{sec:setting2}
We now consider the network model in Fig.~\ref{fig:setting} with the user-to-helper cache association, $\mathcal{U}$, known during the cache placement itself. The prior knowledge of $\mathcal{U}$ helps to design schemes with better performance than the scheme described in Section~\ref{sec:results}. For the setting discussed in Section~\ref{subsec:settingb}, we propose four coded caching schemes, and some of which exhibit optimal performance in certain memory regimes.
\subsection{Scheme 1}
\label{subsec:Scheme1}
The following theorem presents a coded caching scheme that is optimal under uncoded placement. The scheme is based on the MaN scheme in \cite{MaN}, and exists only for certain values of $M_s$ and $M_p$.
\begin{thm}
	Consider a broadcast channel with $K$ users, each having a cache of size $M_p$ files, assisted by $\Lambda$ helper caches, each of normalized size $M_s/N $. For a user-to-helper cache association $\mathcal{U}=\{\mathcal{U}_1, \mathcal{U}_2, \ldots, \mathcal{U}_{\Lambda}\}$ with a profile $\mathcal{L}=(\mathcal{L}_1,\mathcal{L}_2,\ldots,\mathcal{L}_{\Lambda})$, the worst-case rate
	\begin{equation}
	R(M_s,M_p) \Big|_{\textrm{Scheme1}}=\frac{K-t}{t+1}
	\label{eq:rate3}
	\end{equation}
	is achievable if $t \geq \mathcal{L}_1$ and   $M_s \leq \frac{\binom{K-\mathcal{L}_1}{t-\mathcal{L}_1}N}{\binom{K}{t}}$ files, where $t \triangleq K(M_s+M_p)/N$ and $t \in [0,K]$. When $t$ is not an integer, the lower convex envelope of the above points is achievable.
	\label{thm3}
\end{thm}

\begin{IEEEproof}
	The scheme that achieves the rate in \eqref{eq:rate3} is described in the sequel.
	\subsubsection{Placement phase}
	 The total cache size available to each user is $M=M_s+M_p$. Let $t$ be an integer such that $t=KM/N$ and $t\in [0,K]$. The scheme works only if the following conditions are satisfied: $t \geq \mathcal{L}_1$ and $M_s \leq N\binom{K-\mathcal{L}_1}{t-\mathcal{L}_1}/\binom{K}{t}$ files. If the above conditions hold, each file $W_n$ is divided into $\binom{K}{t}$ subfiles. Each subfile is indexed by a $t-$sized subset $\tau \subseteq [K]$. The helper cache placement is done as follows:
	\begin{equation}
	{	\mathcal{Z}_{\lambda}=\{W_{n,\tau}, \forall n \in [N]: \tau \ni \mathcal{U}_{\lambda}(j), \forall j \in [\mathcal{L}_{\lambda}] \}, \textrm{\hspace{0.1cm}} \forall \lambda \in [\Lambda].}
		\label{eq:place_scheme3}
	\end{equation}
	Similarly, the users' private caches are filled as:
	\begin{equation}
	Z_k = \{W_{n, \tau}, \forall n \in [N]: \tau \ni k\}\backslash \mathcal{Z}_{\lambda}, \textrm{\hspace{0.1cm}} \forall k \in [K]
	\label{eq: place_scheme3user}
	\end{equation}
	where $\lambda$ is the helper cache to which the $k^{th}$ user is connected. 
	 If $M_s \leq N\binom{K-\mathcal{L}_1}{t-\mathcal{L}_1}/\binom{K}{t}$ file units, all the subfiles satisfying the condition in \eqref{eq:place_scheme3} cannot be stored fully in the helper caches due to the memory constraint. Hence, the remaining such subfiles are stored in the corresponding user caches according to \eqref{eq: place_scheme3user}. Thus, $|Z_k|=\big(N\binom{K-1}{t-1}/\binom{K}{t}\big)-|\mathcal{Z}_{\lambda}|=M-M_s=M_p$.

	 \textit{2) Delivery phase:}
	 On receiving a demand vector $\mathbf{d}=(d_1,d_2,\ldots,d_K)$, the server generates $\binom{K}{t+1}$ sets $\mathcal{T}\subseteq [K]$, each of size $|\mathcal{T}|=t+1$. For each $\mathcal{T}$, the server sends a message of the form: 
	 %\begin{equation*}
	 $X_{\mathcal{T}}=\underset{k \in \mathcal{T}}{\bigoplus}W_{d_k,\mathcal{T}\backslash k}$.
	 %\end{equation*}
	 
	 The above placement and delivery are identical to the MaN scheme in \cite{MaN} for a dedicated cache network with parameters $K$, $N$, and $M=M_s+M_p$. Hence, the decoding directly follows from it.
	 
	 %	\textit{Decoding:} From the transmissions given, it is straightforward that each user is able to retrieve its desired subfiles from the transmissions using the available cache contents.\\
	\textit{Performance of the scheme:} There are, in total, $\binom{K}{t+1}$ transmissions, each of size $1/\binom{K}{t}$ file units. Thus, the rate achieved by this scheme is $$R(M_s,M_p)\big |_{\textrm{Scheme1}}=\frac{K-t}{t+1}.$$
	This completes the proof of Theorem~\ref{thm3}.
  \end{IEEEproof}	
\textit{Optimality:} The above scheme is optimal under uncoded placement. The optimality follows directly from the fact that the rate achieved in \eqref{eq:rate3} matches exactly with the performance of the optimal scheme for a dedicated cache network \cite{MaN} having a cache size of $M$. Notice that the performance of Scheme $1$ depends only on $M=M_s+M_p$.  Hence, for a fixed $M$, all those $(M_s,M_p)$ pairs satisfying the conditions $M_s+M_p=M$ and $M_s \leq {N\binom{K-\mathcal{L}_1}{t-\mathcal{L}_1}}/{\binom{K}{t}}$ file units achieve the same rate. In Theorem~\ref{thm3}, our main objective is to identify the values of $M_s$ for which the optimal MaN scheme can be employed directly for the network model in Fig.~\ref{fig:setting}.
 
 Now, we consider the case where $t \not\in [0,K]$. When $t=KM/N$ is not an integer, find $M_1={N\floor{t}}/{K}$ and $M_2={N\ceil{t}}/{K}$. When the conditions $M_s \leq N\binom{K-\mathcal{L}_1}{\floor{t}-\mathcal{L}_1}/\binom{K}{\floor{t}}$ and $M_s \leq N\binom{K-\mathcal{L}_1}{\ceil{t}-\mathcal{L}_1}/\binom{K}{\ceil{t}}$ hold, the following memory-sharing technique is used. Let $M_{p_1}$ and $M_{p_2}$ be such that $M_1=M_{s}+M_{p_1}$ and $M_2=M_s+M_{p_2}$. The memory $M$ can be expressed as $M=\alpha M_1+(1-\alpha)M_2$, where $0 \leq \alpha \leq 1 $. Each file $W_n$ is split into $2$ parts, each of size $\alpha$ and $(1-\alpha)$ file units, where $\alpha=1-(t-\floor{t})$.
 Thus, $W_n=\{W^{(\alpha)}_{n},W^{(1-\alpha)}_{n}\} \textrm{\hspace{0.1cm}}, \forall n \in [N]$. The placement corresponding to $\floor{t}$ occupies a total memory of size $\alpha M_1=\alpha(M_s+M_{p_1})$ file units. The file segments $\{W^{(\alpha)}_{n}, \forall n \in [N]\}$ are cached according to $\floor{t}$, and $\alpha R_1$ file units is the transmission length corresponding to it. The remaining memory $(1-\alpha)M_2$ file units is filled according to $\ceil{t}$, and the file segments $\{W^{(1-\alpha)}_{n}, n \in [N]\}$ are  used for it. Corresponding to $\ceil{t}$, the transmission length required is $(1-\alpha)R_2$ file units. Then, the rate $R(M_s,M_p)$ is obtained as
 \begin{align*}
  R(M_s,M_p)&=\alpha R_1 +(1-\alpha)R_2\\ 
 & =\alpha \left(\frac{K-\floor{t}}{\floor{t}+1}\right)+(1-\alpha)\left( \frac{K-\ceil{t}}{\ceil{t}+1} \right).
 \end{align*}
 
   Next, we illustrate the scheme using an example.
   \begin{example}
   $\Lambda=3$, $K=6$, $N=6$, $M_s=6/5$, $M_p=14/5$, $\mathcal{U}=\{\{1,2,3\},\{4,5\},\{6\}\}$ with $\mathcal{L}=(3,2,1)$
   \end{example}
   Consider a case with a server having $N=6$ equal-length files $\mathcal{W}=\{W_1,W_2,\ldots,W_6\}$, connected to $\Lambda=3$ helper caches and to $K=6$ users through an error-free shared link. The helper caches and the user caches are of size $M_s=6/5$ and $M_p=14/5$ file units, respectively. The user-to-helper cache association is $\mathcal{U}=\{\{1,2,3\},\{4,5\},\{6\}\}$.
   
   In the given example, we get $t=4$. Thus, the conditions $t \geq \mathcal{L}_1$ and $M_s \leq N\binom{K-\mathcal{L}_1}{t-\mathcal{L}_1}/\binom{K}{t}=6/5$  are satisfied. Each file, $W_n$, is divided into 15 equally-sized subfiles, and each subfile is indexed using a $4-$sized set $\tau \subseteq [6]$. The contents filled in each helper cache are as follows: $ \mathcal{Z}_1=\{W_{n,1234}, W_{n,1235}, W_{n,1236}, \textrm{\hspace{0.1cm}} \forall n \in [6] \}$, $ \mathcal{Z}_2=\{W_{n,1245}, W_{n,1345}, W_{n,1456}, \textrm{\hspace{0.1cm}} \forall n \in [6]\}$, 
   $\mathcal{Z}_3=\{W_{n,1236}, W_{n,1246}, W_{n,1256}, \textrm{\hspace{0.1cm}} \forall n \in [6]\}$.
 %  \begin{align*}
 %  \mathcal{Z}_1&=\{W_{n,1234}, W_{n,1235}, W_{n,1236}, \textrm{\hspace{0.1cm}} \forall n \in [6] \}, \textrm{\hspace{0.1cm}}
 %  \mathcal{Z}_2=\{W_{n,1245}, W_{n,1345}, W_{n,1456}, \textrm{\hspace{0.1cm}} \forall n \in [6]\},\\
 %  \mathcal{Z}_3&=\{W_{n,1236}, W_{n,1246}, W_{n,1256}, \textrm{\hspace{0.1cm}} \forall n \in [6]\}.
 %  \end{align*}
   Note that the subfiles $W_{n,2345}, W_{n,2456}, W_{n,3456}$, $\forall n \in [6],$ also satisfy the condition in \eqref{eq: place_scheme3user} for $\lambda=2$. But, due to the memory constraint those subfiles cannot be stored in the second helper cache. A similar scenario exists with the third helper cache as well. The contents stored in the users' private cache are as follows:
$$
   \lambda=1
   \begin{cases}
   Z_1= &\{W_{n,1245}, W_{n,1246}, W_{n,1256}, W_{n,1345},
    \\
    & \textrm{\hspace{1cm}}W_{n,1346}, W_{n,1356}, W_{n,1456},  \textrm{\hspace{0.1cm}} \forall n \in [6] \}, \\
   Z_2= &\{W_{n,1245}, W_{n,1246}, W_{n,1256}, W_{n,2345},
     \\
    & \textrm{\hspace{1cm}} W_{n,2346},W_{n,2356}, W_{n,2456},  \textrm{\hspace{0.1cm}} \forall n \in [6] \}, \\
   Z_3= &\{W_{n,1345}, W_{n,1346}, W_{n,1356}, W_{n,2345},
     \\
    & \textrm{\hspace{1cm}}  W_{n,2346},W_{n,2356}, W_{n,3456},  \textrm{\hspace{0.1cm}} \forall n \in [6] \}, 
   \end{cases}
$$
 %  \begin{align*}
 $$  \lambda=2
   \begin{cases}
   Z_4= &\{W_{n,1234}, W_{n,1246}, W_{n,1346}, W_{n,2345},
   \\
    & \textrm{\hspace{1cm}} W_{n,2346}, W_{n,2456}, W_{n,3456},  \textrm{\hspace{0.1cm}} \forall n \in [6] \}, \\
   Z_5= &\{W_{n,1235}, W_{n,1256}, W_{n,1356}, W_{n,2345},
    \\
    &\textrm{\hspace{1cm}} W_{n,2356},W_{n,2456}, W_{n,3456},  \textrm{\hspace{0.1cm}} \forall n \in [6] \}, 
   \end{cases}$$
  % \end{align*}
  % \begin{align*}
$$   \lambda=3
   \begin{cases}
   Z_6= &\{W_{n,1346}, W_{n,1356}, W_{n,1456}, W_{n,2346},
    \\
   & \textrm{\hspace{1cm}} W_{n,2356},W_{n,2456}, W_{n,3456},  \textrm{\hspace{0.1cm}} \forall n \in [6] \}.
   \end{cases}$$
  % \end{align*}
Let $\mathbf{d}=(1,2,3,4,5,6)$ be the demand vector. There is a transmission corresponding to every $5-$sized subset of $[6]$. The transmissions are as follows:
\begin{small}
\begin{align*}
X_{\{12345\}}&=W_{1,2345} \oplus W_{2,1345} \oplus W_{3,1245} \oplus W_{4,1235} \oplus W_{5,1234}, \\
 X_{\{12346\}} & = W_{1,2346} \oplus W_{2,1346} \oplus W_{3,1246} \oplus W_{4,1236} \oplus W_{6,1234}, \\
 X_{\{12356\}} &= W_{1,2356} \oplus W_{2,1356} \oplus W_{3,1256} \oplus W_{5,1236} \oplus W_{6,1235}, \\
 X_{\{12456\}} &= W_{1,2456} \oplus W_{2,1456} \oplus W_{4,1256} \oplus W_{5,1246}  \oplus W_{6,1245}, \\
 X_{\{13456\}} &= W_{1,3456} \oplus W_{3,1456} \oplus W_{4,1356} \oplus W_{5,1346}  \oplus W_{6,1345}, \\
 X_{\{23456\}} &= W_{2,3456} \oplus W_{3,2456} \oplus W_{4,2356} \oplus W_{5,2346} \oplus W_{6,2345}.
\end{align*}
\end{small}
The decoding is straightforward from the transmissions. Each user is able to recover the desired subfiles from the transmissions using the available cache contents. The rate obtained is $R(6/5,14/5)|_{\textrm{Scheme1}}=6/15=2/5$, which is exactly same as the rate achieved by the Maddah-Ali Niesen scheme for a dedicated cache network with $M=4$ files.

Even though Scheme $1$ achieves the optimal performance under uncoded placement, it is limited only to a specific memory regime. Hence, we need to look for schemes that exists for any $(M_s,M_p)$ pair. In the following subsections, we present such schemes.
\subsection{A Scheme from Theorem~\ref{thm1}: Optimizing the rate using $\mathcal{L}$ }
\label{subsec: profile_opt}
When $\mathcal{U}$ is known at the placement, the user-to-cache association agnostic coded caching scheme proposed in Section~\ref{sec:results} can be used  with a change in the value that  $f^{*}$ assumes. Since $\mathcal{U}$ is known a priori, the optimum $f^{*}$ can be found by minimizing $R(M_s,M_p)$ in \eqref{rate} using the actual profile $\mathcal{L}$ instead of the uniform one. Thus, \eqref{eq:f_value1} can be rewritten as
\begin{align}
 f^{*} =\underset{\frac{M_s}{N}\leq f \leq 1-\frac{M_p}{N}} {\arg\min}\textrm{\hspace{0.2cm}}fR_s(t_s)+(1-f)R_p(t_p),
 \label{eq:fopt_known}
\end{align}
 where $R_s(u) = {\sum_{n=1}^{\Lambda-u}\mathcal{L}_n\binom{\Lambda-n}{u}}/{\binom{\Lambda}{u}}$, $u \in [0, \Lambda]$, and $R_p(v)=  {(K-v)}/{(v+1)}$, $v \in [0,K]$. Since $f^{*}$ is obtained by optimizing \eqref{eq:fopt_known} with the exact profile of the network, the rate $R(M_s,M_p)\big|_{\mathcal{U}\textrm{ }known}$ achieved in this case is always at most the rate in \eqref{rate}. The optimality under this placement follows from the converse in the $\mathcal{U}$ known case.

Consider example~\ref{ex:exmp1} again: $N=4$, $K=4$, $\Lambda =2$, $M_s = 2$, $M_p=0.5$, and $\mathcal{L}=(3,1)$. If the actual profile $\mathcal{L}$ is used to minimize the rate, the value of $f^{*}$ obtained is $1/2$. That is, $|W_n^{(s)}|=1/2$ and $|W_n^{(p)}|=1/2$, $\forall n \in [6]$, and we get $t_s = 2$ and $t_p = 1$. Therefore, the rate $R(M_s,M_p)\big|_{\mathcal{L}=(3,1)}$ is obtained as $R(M_s,M_p)\big|_{\mathcal{L}=(3,1)} = f^{*}R_s(2)+(1-f^{*})R_p(1) = 3/4$, which is less than the rate $R(M_s,M_p)\big|_{\textrm{Theorem\ref{thm1}}}=11/12$ obtained in the $\mathcal{U}$ agnostic scheme.
\begin{rem}
	The knowledge of the association profile $\mathcal{L}$ is enough to find the value of $f^{*}$ that minimizes $R(M_s,M_p)$	in the $\mathcal{U}$ unknown case. That is, the user-to-cache association $\mathcal{U}$ is not required to be known to optimize the rate, instead, the association profile $\mathcal{L}$ is sufficient to find $f^{*}$.
\end{rem}
\subsection{Scheme 2}
\label{subsec:scheme2}
Next, we present a scheme that also exists for all $(M_s,M_p)$ pairs, and is optimal in certain memory regimes.
\begin{thm}
	\label{thm2}
	Consider a broadcast channel with $K$ users, each having a cache of size $M_p$ files, assisted by $\Lambda$ helper caches, each of normalized size $M_s/N $. For a user-to-helper cache association $\mathcal{U}=\{\mathcal{U}_1, \mathcal{U}_2, \ldots, \mathcal{U}_{\Lambda}\}$ with a profile $\mathcal{L}=(\mathcal{L}_1,\mathcal{L}_2,\ldots,\mathcal{L}_{\Lambda})$, the following worst-case rate 
	\begin{small}
	\begin{equation}
	\begin{aligned}
	  R(M_s,M_p)\Big|_{\textrm{Scheme2}} = \frac{\displaystyle\sum_{n=1}^{\Lambda-t_s}\binom{\Lambda-n}{t_s}\left[\binom{\mathcal{L}_1}{t_p+1}-\binom{\mathcal{L}_1-\mathcal{L}_n}{t_p+1}\right]}{\binom{\Lambda}{t_s}\binom{\mathcal{L}_1}{t_p}}
	  \label{eq:rate2}
	  \end{aligned}
	\end{equation}
\end{small}
is achievable, where $t_s \triangleq \frac{\Lambda M_s}{N} \in (0,\Lambda]$ and $t_p \triangleq \frac{\mathcal{L}_1 M_p}{N-M_s} \in [0,\mathcal{L}_1]$. Otherwise, the lower convex envelope of the above points is achievable.
\end{thm}
\begin{IEEEproof}
	The achievability of the rate expression in \eqref{eq:rate2} is described below.
	
	\textit{1) Placement phase:} Let the total memory available to each user be denoted as $M$, where $M=M_s+M_p$.
	We first look at the placement in the helper caches. It is same as the placement in \cite{PUE}.
	 Let $t_s \triangleq {\Lambda M_s}/{N}$ such that $t_s \in [\Lambda]$. Each file $W_n \in \mathcal{W}$ is split into $\binom{\Lambda}{t_s}$ number of subfiles, and each subfile is indexed using a $t_s-$sized set $\tau \subseteq [\Lambda]$. The contents stored in the $\lambda^{th}$ helper cache are given as
	 $\mathcal{Z}_{\lambda} = \{ W_{n,\tau}, \forall n \in [N]: \tau \ni \lambda \}.
	 $
	 
	 The users' caches are filled according to the user-to-helper cache association, $\mathcal{U}$. To populate the user caches, each subfile is further split into $\binom{\mathcal{L}_1}{t_p}$ mini-subfiles, where $t_p \triangleq \frac{\mathcal{L}_1 M_p}{N -M_s}$ and $t_p \in [0,\mathcal{L}_1]$. Thus, each mini-subfile is represented as $W_{n,\tau,\rho}$, where $n \in [N], \tau \subseteq [\Lambda] \textrm{ {such that} } |\tau|=t_s$, and $\rho \subseteq [\mathcal{L}_1]$ such that  $|\rho|=t_p$. We assume that the caches and the users associated with them are ordered. Consider a cache $\lambda$ and a user $\mathcal{U}_{\lambda}(j)$, $j \in [\mathcal{L}_{\lambda}]$, connected to it. Then, the contents stored in the private cache of user $\mathcal{U}_{\lambda}(j)$ are given as:
	 \begin{equation}
	  Z_{\mathcal{U}_{\lambda}(j)} = \{ W_{n,\tau,\rho}, \forall n \in [N]: \tau \not\ni \lambda , \rho \ni j  \}.
	  \label{eq:placem}
	 \end{equation}
	 Each user's private cache stores $N\binom{\Lambda-1}{t_s}\binom{\mathcal{L}_1-1}{t_p-1}$ number of mini-subfiles, each of size $\frac{1}{\binom{\Lambda}{t_s}\binom{\mathcal{L}_1}{t_p}}$ file units, thus complying the memory constraint. According to the designed placement policy,
	 for a user $\mathcal{U}_{\lambda}(j)$, we have $\mathcal{Z}_{\lambda} \cap {Z}_{\mathcal{U}_{\lambda}(j)}=\phi$.

	 \textit{2) Delivery phase:} 
	 Let $\mathbf{d}=(d_1,d_2,\ldots,d_K)$ be the demand vector. We define a set $\mathcal{Q}$ as follows: $\mathcal{Q} \triangleq \{\mathcal{S} \times \mathcal{P}: \forall \textrm{\hspace{0.1cm}}\mathcal{S} \subseteq [\Lambda] \textrm{ and }  \mathcal{P} \subseteq [\mathcal{L}_1] \textrm{ such that } |\mathcal{S}|=t_s+1,  |\mathcal{P}|=t_p+1\}$. Corresponding to every $\mathcal{S}\times \mathcal{P} \in \mathcal{Q}$, we find a set of users $\mathcal{U}_{\mathcal{S} \times \mathcal{P}}$ as:
	  \begin{equation}
	 \mathcal{U}_{\mathcal{S} \times \mathcal{P}}= \bigcup_{\lambda \in \mathcal{S}}\big\{ \mathcal{U}_{\lambda}(j): j \in \mathcal{P} \textrm{ and } j \leq \mathcal{L}_{\lambda} \big\}.
	 \end{equation}
	If $\mathcal{U}_{\mathcal{S} \times \mathcal{P}} \neq \phi $, the server transmits a message corresponding to $\mathcal{S}\times \mathcal{P}$ as follows:
	\begin{equation}
	   X_{\mathcal{S}\times \mathcal{P}}=  \underset{\lambda \in \mathcal{S} }{\bigoplus}\Big(\underset{j \in \mathcal{P}: j \leq \mathcal{L}_{\lambda}}{\oplus}W_{d_{\mathcal{U}_{\lambda}(j)}, \mathcal{S}\backslash \lambda, \mathcal{P}\backslash j}\Big).
	\end{equation}
	If $\mathcal{U}_{\mathcal{S} \times \mathcal{P}}=\phi$, there is no transmission corresponding to $\mathcal{S} \times \mathcal{P}$.
	\par
	\textit{Decoding:} Consider a transmission $X_{\mathcal{S}\times \mathcal{P}}$ and a user $\mathcal{U}_{\lambda}(j)$, where $\lambda \in \mathcal{S}$ and $j \in \mathcal{P}$. At user $\mathcal{U}_{\lambda}(j)$, the received message is of the following form
	%\begin{equation}
	\begin{align}
	W_{d_{\mathcal{U}_{\lambda}(j)},\mathcal{S} \backslash \lambda, \mathcal{P} \backslash {j}} \underbrace{\underset{\substack{j^{\prime} \in \mathcal{P}: j^{\prime} \leq \mathcal{L}_{\lambda}\\ j^{\prime} \neq j  }}{\bigoplus} W_{d_{\mathcal{U}_{\lambda}(j^{\prime})},\mathcal{S} \backslash \lambda, \mathcal{P} \backslash j^{\prime}}}_{\text{\textit{term 1}: known from $Z_{\mathcal{U}_{\lambda}(j)}$}}\notag\\ 
	\underbrace{\underset{\lambda^{\prime} \in \mathcal{S}, \lambda \neq \lambda^{\prime} }{\bigoplus} \Big( \underset{\tilde{j} \in \mathcal{P}: \tilde{j} \leq \mathcal{L}_{\lambda^{\prime}}}{\oplus}     W_{d_{\mathcal{U}_{\lambda^{\prime}}(\tilde{j})}, \mathcal{S}\backslash \lambda^{\prime}, \mathcal{P}\backslash \tilde{j}} \Big)}_{\text{\textit{term 2}: known from $\mathcal{Z}_{\lambda}$}}.
	\label{eq:dec2}
	\end{align}
	%\end{equation}
	The user $\mathcal{U}_{\lambda}(j)$ can cancel \textit{term 1} and \textit{term 2} from \eqref{eq:dec2} as the subfiles composing \textit{term 1} and \textit{term 2} are available in its private and shared caches, respectively.
	\par
	\textit{Performance of the scheme}: In the delivery scheme, the server first generates a set $\mathcal{Q}$ which is formed by the cartesian product of all possible $(t_s+1)-$sized subsets of $[\Lambda]$ and $(t_p+1)-$sized subsets of $[\mathcal{L}_1]$. Corresponding to every set $\mathcal{S} \times \mathcal{P} \in \mathcal{Q}$, a set of receiving users $\mathcal{U}_{\mathcal{S} \times \mathcal{P}}$ is found. The number of transmissions is determined by the number of non-empty sets, $\mathcal{U}_{\mathcal{S} \times \mathcal{P}}$. Since the caches are labelled in the non-increasing order of the number of users connected to it, the number of non-empty $\mathcal{U}_{\mathcal{S} \times \mathcal{P}}$ is calculated as follows: each set $\mathcal{S} \times \mathcal{P}$ with $\mathcal{S} \ni \lambda=1$ results in $\binom{\mathcal{L}_1}{t_p+1}$ number of non-empty sets, $\mathcal{U}_{\mathcal{S} \times \mathcal{P}}$. For sets $\mathcal{S} \times \mathcal{P}$ with $\mathcal{S} \not\ni \lambda=1 $, define $\lambda_{\mathcal{S}} \triangleq \min\{\lambda: \lambda \in \mathcal{S}\}$. For each such $\mathcal{S} \times \mathcal{P}$, there are $\binom{\mathcal{L}_1-\mathcal{L}_{\lambda_{\mathcal{S}}}}{t_p+1}$ number of empty sets, $\mathcal{U}_{\mathcal{S}\times \mathcal{P}}$. Thus, the total  number of non-empty $\mathcal{U}_{\mathcal{S} \times \mathcal{P}}$ is given as $\sum_{n=1}^{\Lambda}\binom{\Lambda-n}{t_s}\big[\binom{\mathcal{L}_1}{t_p+1}-\binom{\mathcal{L}_1-\mathcal{L}_n}{t_p+1}\big]$. Since each transmission is of size $1/[\binom{\Lambda}{t_s}\binom{\mathcal{L}_1}{t_p}]$ file units, we obtain $R(M_s,M_p)|_{\textrm{Scheme2}}$ as
	\begin{equation*}
	 \small{ \frac{\displaystyle\sum_{n=1}^{\Lambda-t_s}\binom{\Lambda-n}{t_s}\left[\binom{\mathcal{L}_1}{t_p+1}-\binom{\mathcal{L}_1-\mathcal{L}_n}{t_p+1}\right]}{\binom{\Lambda}{t_s}\binom{\mathcal{L}_1}{t_p}}.}
	\end{equation*}
	This completes the proof of Theorem~\ref{thm2}.
\end{IEEEproof}
When $t_s=0$, the setting is equivalent to a dedicated cache network model. Therefore, we can follow the MaN scheme in \cite{MaN} for such cases.

\begin{corollary}
	\label{cor3}
	When the association of users to helper caches is uniform, the rate achieved is 
\begin{equation}
 R(M_s,M_p)\big|_{\textrm{Scheme2}}=\frac{K\big(1-M/{N}\big)}{(t_s+1)(t_p+1)},
\end{equation}
where $M=M_s+M_p$.
\end{corollary}
\begin{IEEEproof}
	For a uniform user-to-helper cache association, the profile is $\mathcal{L}=(\frac{K}{\Lambda},\ldots,\frac{K}{\Lambda})$. Then, the rate expression in \eqref{eq:rate2} becomes
	%\begin{subequations}
	\begin{align}
	 R(M_s,M_p)\big|_{\textrm{Scheme2}}&= \frac{\sum_{n=1}^{\Lambda-t_s}\binom{\Lambda-n}{t_s}\binom{\mathcal{L}_1}{t_p+1}}{\binom{\Lambda}{t_s}\binom{\mathcal{L}_1}{t_p}}\notag\\
	& = \frac{\binom{\Lambda}{t_s+1}\binom{\mathcal{L}_1}{t_p+1}}{\binom{\Lambda}{t_s}\binom{\mathcal{L}_1}{t_p+1}}=\frac{(\Lambda-t_s)(\mathcal{L}_1-t_p)}{(t_s+1)(t_p+1)}. \label{eq:uniform}
	\end{align}
%	\end{subequations}
Substituting for $\mathcal{L}_1$ and further simplifying \eqref{eq:uniform}, we get
\begin{equation*}
R(M_s,M_p)\Big|_{\textrm{Scheme2}}=\frac{K \big( 1-M/N \big)}{(t_s+1)(t_p+1)}.
\end{equation*}
This completes the proof of Corollary~\ref{cor3}. 
\end{IEEEproof}		

\begin{rem}
	For a uniform user-to-helper cache association, the coding gain, which is defined as the number of users benefiting from a transmission, is $(t_s+1)(t_p+1)$ for all transmissions.	
\end{rem}

When $t_s \not\in (0,\Lambda]$ or $t_p \not\in [0,\mathcal{L}_1]$, we need to employ the memory-sharing technique described below. 

%\begin{enumerate}
	 1) First, consider a case when $t_s=\Lambda M_s/N \in (0,\Lambda]$ and $t_p ={\mathcal{L}_1M_p}/{(N-M_s)} \not\in [0,\mathcal{L}_1]$. The helper cache placement is same as described in the scheme. Obtain $M_{p_1}=(N-M_s)\floor{t_p}/\mathcal{L}_1$ and $M_{p_2}=(N-M_s)\ceil{t_p}/\mathcal{L}_1$. Then, $M_p=\alpha M_{p_1}+(1-\alpha)M_{p_2}$, where $0 \leq \alpha \leq 1$. Each subfile in the set $\{W_{n,\tau}, \forall n \in [N]: \tau \subseteq [\Lambda], |\tau|=t_s\}$ is segmented into two parts: $W^{(\alpha)}_{n,\tau}$ and $W^{(1-\alpha)}_{n,\tau}$, each of size $\alpha/\binom{\Lambda}{t_s}$ and $(1-\alpha)/\binom{\Lambda}{t_s}$ file units, respectively. %The users' private caches are also divided into two parts, each of size $\alpha M_{p_1}$ and $(1-\alpha) M_{p_2}$ file units.
	  The mini-subfiles of $\{W^{(\alpha)}_{n,\tau}, \forall n \in [N]: \tau  \subseteq [\Lambda], |\tau|=t_s\}$ are cached according to \eqref{eq:placem} with $|\rho|=\floor{t_p}$, occupying $\alpha M_{p_1}$ portion of each user's private cache. The transmission length required with $t_s$ and $\floor{t_p}$ is $\alpha R_1$ file units. Likewise, the placement of the mini-subfiles of $\{W^{(1-\alpha)}_{n,\tau}, \forall n \in [N]: \tau  \subseteq [\Lambda], |\tau|=t_s\}$ is done with $|\rho|=\ceil{t_p}$, and  $(1-\alpha)R_2$ file units is the transmission length with $t_s$ and $\ceil{t_p}$. Then, the rate $R(M_s,M_p)$ is obtained as $R(M_s,M_p)=\alpha R_1(M_s,M_{p_1})+(1-\alpha)R_2(M_s,M_{p_2})$.
	
	2) Next, assume $t_s = \Lambda M_s/N \not \in (0,\Lambda]$. In this case, it is not required to check whether $t_p$ is an integer or not. Find $M_{s_1}=N\floor{t_s}/\Lambda$ and $M_{s_2}=N\ceil{t_s}/\Lambda$. Then, $M_s =\alpha M_{s_1}+(1-\alpha)M_{s_2}$, where $0 \leq \alpha \leq 1$. Divide the entire file library into two parts: $\{W^{(\alpha)}_{n}, \forall n \in [N]\}$ and $\{W^{(1-\alpha)}_{n}, \forall n \in [N]\}$ such that $|W^{(\alpha)}_{n}|=\alpha$ file units and $|W^{(1-\alpha)}_{n}|=(1-\alpha)$ file units. Using $M_{s_1}$, obtain $t_{p_1}=\mathcal{L}_1 M_p/(N-M_{s_1})$. If $t_{p_1}$ is an integer, the placement and delivery are done according to Scheme $2$ with $\floor{t_s}$ and $t_{p_1}$. The helper cache placement with $\floor{t_s}$ is performed using the set $\{W^{(\alpha)}_{n}, \forall n \in [N]\}$, and it occupies only $\alpha M_{s_1}$ portion of the helper caches. If $t_{p_1} \not\in [0,\mathcal{L}_1]$, we need to follow the procedure mentioned in the first case by finding $M_{p_{1,1}}=(N-M_{s_1})\floor{t_{p_1}}/\mathcal{L}_1$ and $M_{p_{1,2}}=(N-M_{s_1})\ceil{t_{p_1}}/\mathcal{L}_1$. Then, $M_{p}=\beta M_{p_{1,1}}+(1-\beta)M_{p_{1,2}}$, where $0 \leq \beta \leq 1$. Each subfile in the set $\{W^{(\alpha)}_{n,\tau}, \forall n \in [N]: \tau \subseteq [\Lambda], |\tau|=\floor{t_{s}}\}$ is divided into two parts: $W^{(\alpha\beta)}_{n,\tau}$ and $W^{(\alpha-\alpha\beta)}_{n,\tau}$, each with size $\alpha\beta/\binom{\Lambda}{\floor{t_{s}}}$ file units and $\alpha(1-\beta)/\binom{\Lambda}{\floor{t_{s}}}$ file units, respectively. The portion of each user's private cache occupied by the mini-subfiles of $\{W^{(\alpha\beta)}_{n,\tau}, \forall n \in [N]: \tau \subseteq [\Lambda], |\tau|=\floor{t_{s}}\}$ having $|\rho|=\floor{t_{p_1}}$ is $\alpha \beta M_{p_{1,1}}$. The transmission length obtained in this case is $\alpha \beta R_{1,1}$ file units. Similarly, $\alpha(1-\beta)M_{p_{1,2}}$ is the portion of the memory occupied by the mini-subfiles of the set $\{W^{(\alpha-\alpha\beta)}_{n,\tau}, \forall n \in [N]: \tau \subseteq [\Lambda], |\tau|=\floor{t_{s}}\}$ with $|\rho|=\ceil{t_{p_1}}$, and $\alpha(1-\beta)R_{1,2} $ file units is the corresponding transmission length required. 
	 
	 The remaining $(1-\alpha)M_{s_2}$ portion of each helper cache is filled with the subfiles of the set $\{W^{(1-\alpha)}_{n}, \forall n \in [N]\}$ using the value $\ceil{t_s}$. Then, find $t_{p_2}=\mathcal{L}_1M_p/(N-M_{s_2})$. If $t_{p_2}$ is an integer, the placement and delivery are done using Scheme $2$ with $\ceil{t_s}$ and $t_{p_2}$. If $t_{p_2}$ is not an integer, find $M_{p_{2,1}}=(N-M_{s_2})\floor{t_{p_2}}/\mathcal{L}_1$ and $M_{p_{2,2}}=(N-M_{s_2})\ceil{t_{p_2}}/\mathcal{L}_1$. Then, $M_p=\gamma M_{p_{2,1}}+(1-\gamma)M_{p_{2,2}}$ such that $0 \leq \gamma \leq 1$. Each subfile in the set $\{W^{(1-\alpha)}_{n,\tau}, \forall n \in [N]: \tau \subseteq [\Lambda], |\tau|=\ceil{t_s}\}$ is partitioned into $W^{(\gamma -\gamma\alpha)}_{n,\tau}$ and $W^{((1 -\alpha)(1-\gamma))}_{n,\tau}$, where $|W^{(\gamma -\gamma\alpha)}_{n,\tau}|=(1-\alpha)\gamma/\binom{\Lambda}{\ceil{t_s}}$ file units and $|W^{((1 -\alpha)(1-\gamma))}_{n,\tau}|=(1-\alpha)(1-\gamma)/\binom{\Lambda}{\ceil{t_s}}$ file units. The remaining $(1-\alpha)M_p$ portion of each user's private caches is filled according to \eqref{eq:placem} with the mini-subfiles of the following sets: $\{W^{(\gamma-\gamma\alpha)}_{n,\tau}, \forall n \in [N]: \tau \subseteq [\Lambda], |\tau|=\ceil{t_{s}}\}$ with $|\rho|=\floor{t_{p_2}}$ and $\{W^{((1-\alpha)(1-\gamma))}_{n,\tau}, \forall n \in [N]: \tau \subseteq [\Lambda], |\tau|=\ceil{t_{s}}\}$ with $|\rho|=\ceil{t_{p_2}}$. The mini-subfiles having $|\rho|=\floor{t_{p_2}}$ occupies a memory of $(1-\alpha)\gamma M_{p_{2,1}}$, and the remaining $(1-\alpha)(1-\gamma)M_{p_{2,2}}$ portion is occupied by the mini-subfies having $|\rho|=\ceil{t_{p_2}}$. The transmission length required obtained in the case of $\ceil{t_s}$ and $\floor{t_{p_2}}$ is $(1-\alpha)\gamma R_{2,1}$ file units, and $(1-\alpha)(1-\gamma)R_{2,2}$ file units is the transmission length required for $\ceil{t_s}$ and $\ceil{t_{p_2}}$.
	 
	 Thus, the given memory pair $(M_s,M_p)$ is obtained as the convex linear combination of four other points $(M_{s_1},M_{p_{1,1}})$, $(M_{s_1},M_{p_{1,2}})$,  $(M_{s_2},M_{p_{2,1}})$, and $(M_{s_2},M_{p_{2,2}})$ as:
	 \begin{align*}
	(M_s,&M_p)= \alpha\beta(M_{s_1},M_{p_{1,1}})+\alpha(1-\beta)(M_{s_1},M_{p_{1,2}})+\\&(1-\alpha)\gamma(M_{s_2},M_{p_{2,1}})+(1-\alpha)(1-\gamma)(M_{s_2},M_{p_{2,2}}).
	 \end{align*}
	 Therefore, the rate $R(M_s,M_p)$ is obtained as: 
	 \begin{small}
	 \begin{align*} 
	  R(M_s,M_p)=\alpha\beta R_{1,1}(M_{s_1},M_{p_{1,1}})+ \alpha(1-\beta)R_{1,2}(M_{s_1},M_{p_{1,2}})+\\(1-\alpha)\gamma R_{2,1}(M_{s_2},M_{p_{2,1}})+ (1-\alpha)(1-\gamma)R_{2,2}(M_{s_2},M_{p_{2,2}}).
	  \end{align*}  
  \end{small}
%\end{enumerate}
Next, we present an example to illustrate the scheme.
\begin{example}
  $\Lambda=3$, $K=6$, $N=6$, $M_s=2$, $M_p=4/3$, $\mathcal{U}=\{\{1,2,3\},\{4,5\},\{6\}\}$ with $\mathcal{L}=(3,2,1)$	
  \label{ex:ex1_scheme2}
\end{example}
Consider a scenario where there is a server with $N=6$ files, $\mathcal{W}=\{W_1,W_2,\ldots,W_6\}$, connected to $\Lambda=3$ helper caches and to $K=6$ users through a wireless broadcast link. Each helper cache and user cache are of size $M_s=2$ files and $M_p=4/3$ files, respectively.

In this example, $t_s=1$ and $t_p=1$. To fill the helper caches, each file $W_n, n \in [6]$, is divided into $3$ subfiles $\{W_{n,1},W_{n,2},W_{n,3}\}$. The contents cached at each helper cache are as follows:
\begin{align*}
\mathcal{Z}_1=\{W_{n,1}, \forall n \in [6]\},& \mathcal{Z}_2=\{W_{n,2}, \forall n \in [6]\},\\ \textrm{ and }
\mathcal{Z}_3=&\{W_{n,3}, \forall n \in [6]\}. 
\end{align*}

To fill the user caches, each subfile is further divided into $3$ mini-subfiles $\{W_{n,\tau,1}, W_{n,\tau,2}, W_{n,\tau,3}\}$, where $n \in [6]$ and $\tau \in [3]$. The contents stored in each user cache are given below.
\begin{comment}
\begin{align*}
Z_1&=\{W_{n,2,1}, W_{n,3,1}, \textrm{\hspace{0.1cm}} \forall n \in [6]\},\\
Z_2&=\{W_{n,2,2}, W_{n,3,2}, \textrm{\hspace{0.1cm}}\forall n \in [6]\}, \\
Z_3&=\{W_{n,2,3}, W_{n,3,3}, \textrm{\hspace{0.1cm}}\forall n \in [6]\}, \\
Z_4&=\{W_{n,1,1,}, W_{n,3,1}, \textrm{\hspace{0.1cm}}\forall n \in [6]\},\\
Z_5&=\{W_{n,1,2}, W_{n,3,2}, \textrm{\hspace{0.1cm}}\forall n \in [6]\}, \\ Z_6&=\{W_{n,1,1},W_{n,2,1}, \textrm{\hspace{0.1cm}}\forall n \in [6]\}.
\end{align*}
\end{comment}
$$
\lambda=1
\begin{cases}
Z_1=\{W_{n,2,1}, W_{n,3,1}, \textrm{\hspace{0.1cm}} \forall n \in [6]\},\textrm{\hspace{0.1cm}}\\
Z_2=\{W_{n,2,2}, W_{n,3,2}, \textrm{\hspace{0.1cm}}\forall n \in [6]\},\\
Z_3=\{W_{n,2,3}, W_{n,3,3}, \textrm{\hspace{0.1cm}}\forall n \in [6]\}, 
\end{cases}
$$
$$
\lambda=2
\begin{cases}
Z_4=\{W_{n,1,1,}, W_{n,3,1}, \textrm{\hspace{0.1cm}}\forall n \in [6]\},\textrm{\hspace{0.1cm}}\\
Z_5=\{W_{n,1,2}, W_{n,3,2}, \textrm{\hspace{0.1cm}}\forall n \in [6]\}, 
\end{cases}
$$
$$
%\hspace{-2.5cm}
\lambda=3
\begin{cases}
Z_6=\{W_{n,1,1},W_{n,2,1}, \textrm{\hspace{0.1cm}}\forall n \in [6]\}.
\end{cases}
$$

Let $\mathbf{d}=(1,2,3,4,5,6)$ be the demand vector. The set $\mathcal{Q}$ is obtained as follows: $\mathcal{Q}=\big\{\{12,12\}, \{12,13\}, \{12,23\},$ $\{13,12\}, \{13,13\}, \{13,23\},\{23,12\}, \{23,13\},\{23,23\}\big\}$. Then, the set of users $\mathcal{U}_{\mathcal{S} \times \mathcal{P}}$, $\forall \mathcal{S} \times \mathcal{P} \in \mathcal{Q}$, is obtained as follows:
%\begin{align*}
$\mathcal{U}_{\{12,12\}}=\{1,2,4,5\},\textrm{\hspace{0.1cm}} \mathcal{U}_{\{12,13\}}=\{1,3,4\}, \textrm{\hspace{0.1cm}}
\mathcal{U}_{\{12,23\}}=\{2,3,5\}, \textrm{\hspace{0.1cm}} \mathcal{U}_{\{13,12\}}=\{1,2,6\}, \textrm{\hspace{0.1cm}}
 \mathcal{U}_{\{13,13\}}=\{1,3,6\}, \textrm{\hspace{0.1cm}} \mathcal{U}_{\{13,23\}}=\{2,3\}, \textrm{\hspace{0.1cm}}$ 
  $\mathcal{U}_{\{23,12\}}=\{4,5,6\}, \textrm{\hspace{0.1cm}} \mathcal{U}_{\{23,13\}}=\{4,6\}, \textrm{\hspace{0.1cm}} \mathcal{U}_{\{23,23\}}=\{5\}.$
%\end{align*}
Since $\mathcal{U}_{\mathcal{S} \times \mathcal{P}}\neq \phi$, $\forall \mathcal{S} \times \mathcal{P} \in \mathcal{Q}$, the server sends a message corresponding to every $\mathcal{S} \times \mathcal{P}$. Note that the number of users benefiting from each transmitted message is not the same. The transmissions are as follows:
\begin{align*}
X_{\{12,12\}}&=W_{1,2,2} \oplus  W_{2,2,1} \oplus W_{4,1,2} \oplus W_{5,1,1},\textrm{\hspace{0.1cm}}\\
X_{\{12,13\}}&=W_{1,2,3} \oplus W_{3,2,1} \oplus W_{4,1,3},\\
X_{\{12,23\}}&=W_{2,2,3} \oplus W_{3,2,2} \oplus W_{5,1,3},\textrm{\hspace{0.1cm}}\\
X_{\{13,12\}}&=W_{1,3,2} \oplus W_{2,3,1} \oplus W_{6,1,2}, \\
X_{\{13,13\}}&=W_{1,3,3} \oplus W_{3,3,1}   \oplus W_{6,1,3}, \textrm{\hspace{0.1cm}}\\
X_{\{13,23\}}&=W_{2,3,3} \oplus W_{3,3,2}, \\
X_{\{23,12\}}&=W_{4,3,2} \oplus W_{5,3,1} \oplus W_{6,2,2},\textrm{\hspace{0.1cm}}\\
X_{\{23,13\}}&=W_{4,3,3} \oplus W_{6,2,3},\textrm{\hspace{0.1cm}}
X_{\{23,23\}}=W_{5,3,3}.
\end{align*}

To explain the decoding, consider user $1$. The transmissions beneficial to user $1$ are $X_{\{12,12\}}$, $X_{\{12,13\}}$, $X_{\{13,12\}}$, and $X_{\{13,13\}}$. Consider transmission $X_{\{12,12\}}$. The subfile $W_{2,2,1}$ is cached in user $1$'s private cache, and the other two subfiles, $W_{\{4,1,2\}}$ and $W_{\{5,1,1\}}$, are available from the helper cache to which it is connected. Thus, user $1$ can decode the desired subfile $W_{1,2,2}$. Similarly, it can decode the remaining subfiles. The procedure is same for all other users as well.
Thus, the rate achieved is $R(2,4/3)|_{\textrm{Scheme2}}=9/9=1$.

 Now, let us look at the rate achieved by the scheme in Section~\ref{subsec: profile_opt} in this case. With $\mathcal{U}$ known, we get $R(2,4/3)\big|_{\mathcal{L}=(3,1)}=0.89$ which is less than the rate achieved by Scheme $2$. However, there are instances where Scheme $2$ exhibits a better performance than the scheme in Section~\ref{subsec: profile_opt}. One such instance is described below.
 
\begin{example}
 	$N=4$, $K=4$, $\Lambda=2$, $M_s=2$, $M_p=0.5$, $\mathcal{U}=\{\{1,2\},\{3,4\}\}$ with $\mathcal{L}=(2,2)$
 	\label{ex:ex2_scheme2}
\end{example}
 
 Consider a network with a server having $N= 4$ unit-sized files, $\mathcal{W}=\{W_1,W_2,W_3,W_4\}$. The server is connected to $\Lambda = 2$ helper caches, each of size $M_s = 2$ units, and to $K=4$ users, each having a dedicated cache of size $M_p=0.5$ units. The users are associated with the caches in a uniform manner as given: $\mathcal{U}=\{\{1,2\},\{3,4\}\}$.

In Scheme $2$, we get $t_s=1$ and $t_p=1/2$. To fill the helper caches, each file $W_n, n \in [4]$,l is divided into equally-sized subfiles: $W_n =\{W_{n,1},W_{n,2}\}$. The helper cache placement is as follows: $\mathcal{Z}_1 = \{W_{n,1},\forall n \in [4]\}$ and $\mathcal{Z}_2 = \{W_{n,2},\forall n \in [4]\}$. Since $t_p \not\in [2]$, the memory sharing technique needs to be employed between $\floor{t_p}$ and $\ceil{t_p}$ instances. Corresponding to $\floor{t_p}$ and $\ceil{t_p}$, we obtain $M_{p_1} =  0$ and $M_{p_2} = 1$. Thus, $M_p = \alpha M_{p_1} + (1-\alpha)M_{p_2}$, where $\alpha = 1/2$. Each subfile $W_{n,\tau}$, $n \in [4]$, $\tau \in [2]$, is further divided into two parts as $W_{n,\tau} = \{W_{n,\tau}^{(\alpha)},W_{n,\tau}^{(1-\alpha)}\}$, where $|W_{n,\tau}^{(\alpha)}|=1/4$ units and $|W_{n,\tau}^{(1-\alpha)}|=1/4$ units. 

 The case $\floor{t_p}=0$ corresponds to a fully shared cache system, and hence, Scheme $2$ reduces to the optimal shared cache scheme in \cite{PUE}. Corresponding to $\ceil{t_p}=1$, each subfile in the set $\{W_{n,\tau}^{(1-\alpha)}, \forall \tau \in [2], n \in [4]\}$ is further divided into two mini-subfiles  $W_{n,\tau,1}^{(1-\alpha)}$ and $W_{n,\tau,2}^{(1-\alpha)}$, each of size $1/8$ units.
 Thus, the contents stored at the users' caches are as follows:
 \begin{align*}
 Z_1 & = \{W^{(1-\alpha)}_{n,2,1} \forall n \in [4]\},\textrm{\hspace{0.3cm}} Z_2 = \{W^{(1-\alpha)}_{n,2,2} \forall n \in [4]\},\\
 Z_3 & = \{W^{(1-\alpha)}_{n,1,1} \forall n \in [4]\}, \textrm{\hspace{0.3cm}} Z_4 = \{W^{(1-\alpha)}_{n,1,2} \forall n \in [4]\}.
 \end{align*}
 
 Let $\mathbf{d}=(1,2,3,4)$ be the demand vector. Then, the transmissions are as follows:
 \begin{align*}
   X_{\{12\},1}^{(\alpha)} &= W^{(\alpha)}_{1,2} \oplus  W^{(\alpha)}_{3,1}, \textrm{\hspace{0.2cm}} X_{\{12\},2}^{(\alpha)} = W^{(\alpha)}_{2,2} \oplus  W^{(\alpha)}_{4,1},\\
   X_{\{12,12\}}^{(1-\alpha)} & = W^{(1-\alpha)}_{1,2,2} \oplus W^{(1-\alpha)}_{2,2,1} \oplus W^{(1-\alpha)}_{3,1,2} \oplus W^{(1-\alpha)}_{4,1,1}.
 \end{align*}
 Each user is able to retrieve its demanded file using the above transmissions and the available cache contents. Thus, the rate obtained in this case is $R(2,0.5)\big|_{\textrm{Scheme2}} = 1/2 + 1/8 = 5/8$. Whereas, the rate achieved by the scheme in Section~\ref{subsec: profile_opt} in this case is $R(2,0.5)\big|_{\mathcal{L}=(2,2)} = 2/3$ (the value of $f^{*}$ obtained is $3/4$) which is greater than the rate achieved by  Scheme $2$. Later in Section~\ref{sec:lb}, we show that Scheme $2$ is optimal in certain memory regimes.

 From the previous two examples (Examples~\ref{ex:ex1_scheme2} and~\ref{ex:ex2_scheme2}), we have seen that neither Scheme $2$ nor the scheme in Section~\ref{subsec: profile_opt} performs well in all memory regimes. Hence, we propose another scheme in the following subsection that leverages the advantages provided by both the schemes and result in a better performance when $\mathcal{U}$ is known beforehand.
 
 \subsection{Composite scheme: Combination of MaN scheme and Scheme $2$}
\label{subsec:MAN_scheme2}
We have seen, so far, three different schemes for the case when $\mathcal{U}$ is known at the placement phase itself. Except Scheme $1$, the other two schemes exist in all memory regimes and each of them offers advantage in different memory regions. Hence, we design a new scheme that encompasses the advantages offered by the previous two schemes (Scheme in Section~\ref{subsec: profile_opt} and Scheme $2$) into a single one. We refer this scheme as composite scheme as it combines MaN scheme  and Scheme 2. The composite scheme is based on partitioning the files and the users' private memory appropriately to share them between MaN scheme and Scheme $2$.

 The scheme is as follows: divide each file $W_n,  n \in [N]$, into two non-overlapping parts: $W_n = \{W_n^{(M)},W_n^{(S2)}\} $ such that $|W_n^{(M)}|=z$ file units and $|W_n^{(S2)}|=1-z$ file units, where $0 \leq z \leq 1$. The set $\{W_n^{(M)}, \forall n \in [N]\}$ is used to fill $vM_p$ portion of the user's private cache using  MaN  scheme placement, where $0 \leq v \leq 1$. The helper caches and the remaining $(1-v)M_p$ portion of each user's private cache are filled according to Scheme $2$ using the set $\{W_n^{(S2)}, \forall n \in [N]\}$. On receiving a demand vector $\mathbf{d}$, the server independently runs the two delivery policies corresponding to MaN scheme and Scheme $2$ on $\{W_n^{(M)}, \forall n \in [N]\}$ and $\{W_n^{(S2)}, \forall n \in [N]\}$, respectively. We define, $t_p^{(M)} \triangleq \frac{KvM_p}{zN} \in [0,K]$, $t_s^{(S2)} \triangleq \frac{\Lambda M_s}{(1-z)N}\in (0,\Lambda]$, and $t_p^{(S2)} \triangleq \frac{\mathcal{L}_1(1-v) M_p}{(1-z)N-M_s} \in [0,\mathcal{L}_1]$. When $t_p^{(M)}$, $t_s^{(S2)}$ and $t_p^{(S2)}$ are integers, the rate achieved by the composite scheme is:
 \begin{align*}
   R&(M_s,M_p)\big|_{\substack{\textrm{Composite} \\ \textrm{scheme}}} =z\left(\frac{K-t_p^{(M)}}{t_p^{(M)}+1}\right)+\\&(1-z)\left( \frac{\displaystyle\sum_{n=1}^{\Lambda-t_s}\binom{\Lambda-n}{t_s^{(S2)}}\left[\binom{\mathcal{L}_1}{t_p^{(S2)}+1}-\binom{\mathcal{L}_1-\mathcal{L}_n}{t_p^{(S2)}+1}\right]}{\binom{\Lambda}{t_s^{(S2)}}\binom{\mathcal{L}_1}{t_p^{(S2)}}}\right).
 \end{align*}
 If any of the parameters $t_p^{(M)}$, $t_s^{(S2)}$, and $t_p^{(S2)}$ are not integers, we need to employ the memory-sharing technique. The parameters $z$ and $v$ are obtained by optimizing $R(M_s,M_p)\big|_{\substack{\textrm{Composite} \\ \textrm{scheme}}}$ over $0 \leq v \leq 1$ and $\frac{vM_p}{N} \leq z \leq 1-\frac{M_s}{N}$. The optimal $z^{*}$ and $v^{*}$ obtained through the optimization is then used in the scheme.
%where $t_p^{(M)} \triangleq \frac{KvM_p}{zN} \in [0,K]$, $t_s^{(S2)} \triangleq \frac{\Lambda M_s}{(1-z)N}\in (0,\Lambda]$, and $t_p^{(S2)} \triangleq \frac{\mathcal{L}_1(1-v) M_p}{(1-z)N-M_s} \in [0,\mathcal{L}_1]$.
\begin{rem}
 When $v=1$, the composite scheme subsumes the scheme in Section~\ref{subsec: profile_opt} as a special case.
\end{rem}
Note that the rate achieved by the composite scheme will always be at most the rate achieved   either by Scheme $2$ or by scheme in Section~\ref{subsec: profile_opt}. That is, $R(M_s,M_p)\big|_{\substack{\textrm{Composite} \\ \textrm{scheme}}} \leq \min\{R(M_s,M_p)\big|_{\mathcal{L}},R(M_s,M_p)\big|_{\textrm{Scheme2}}\}$, where $R(M_s,M_p)\big|_{\mathcal{L}}$ denotes the rate-memory tradeoff of the scheme in  Section~\ref{subsec: profile_opt}. Hence, it is enough to consider only the composite scheme while doing a performance comparison.

\section{Numerical Comparisons}
\label{sec:analyses}
In this section, we compare and characterize the performances of the proposed schemes in Theorem~\ref{thm1} ($\mathcal{U}$ unknown case), Theorem~\ref{thm3} (Scheme $1$), and in Section~\ref{subsec:MAN_scheme2} (composite scheme). We consider a scenario with $N=30$ files, $K=30$ users, and $\Lambda=5$ helper caches. In Fig.~\ref{fig:skewed}, the plots are drawn with a $\mathcal{U}=\{\{1,2,\ldots,20\},\{21,\ldots,24\},$ $\{25,26,27\},\{28,29\},\{30\}\}$ having the profile $\mathcal{L}=(20,4,3,2,1)$. The plots in Fig.~\ref{fig:skewed} are drawn by fixing $M_s$, and varying $M_p$ such that $M_s+M_p \leq N$. The rate-memory curves $R^{*}_{\textrm{PUE}}(M)$ and $R^{*}_{\textrm{MaN}}(M)$, where $M=M_s+M_p$, serve as an upper bound and a lower bound, respectively, for the proposed schemes.

\begin{figure*}
	\centering
	\begin{subfigure}{0.45\textwidth}
		\includegraphics[width=\textwidth]{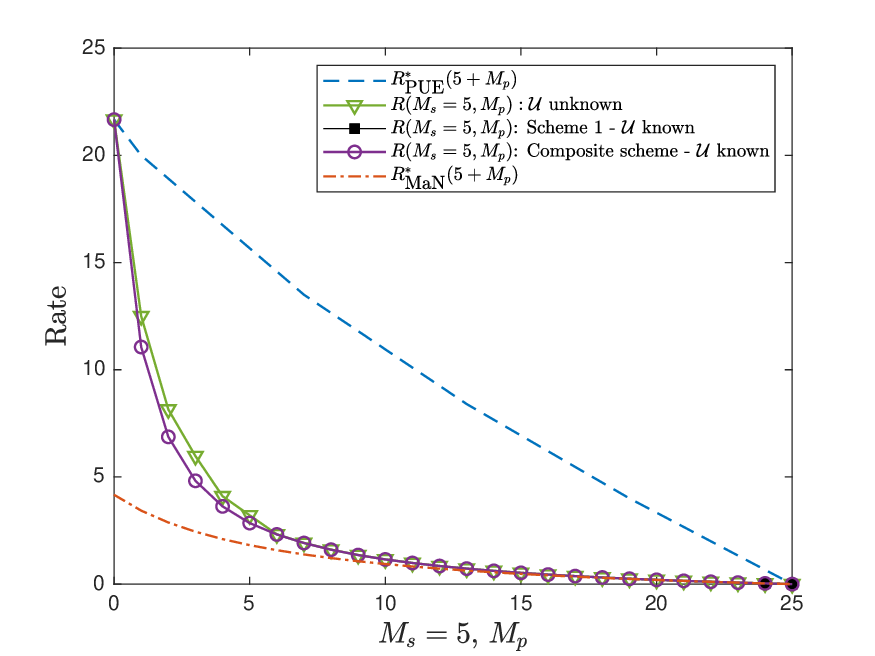}
		\caption{$M_s=5$}	
		\label{fig:ts1skewed}
	\end{subfigure}
	\hfill
	\begin{subfigure}{0.45\textwidth}
		\includegraphics[width=\textwidth]{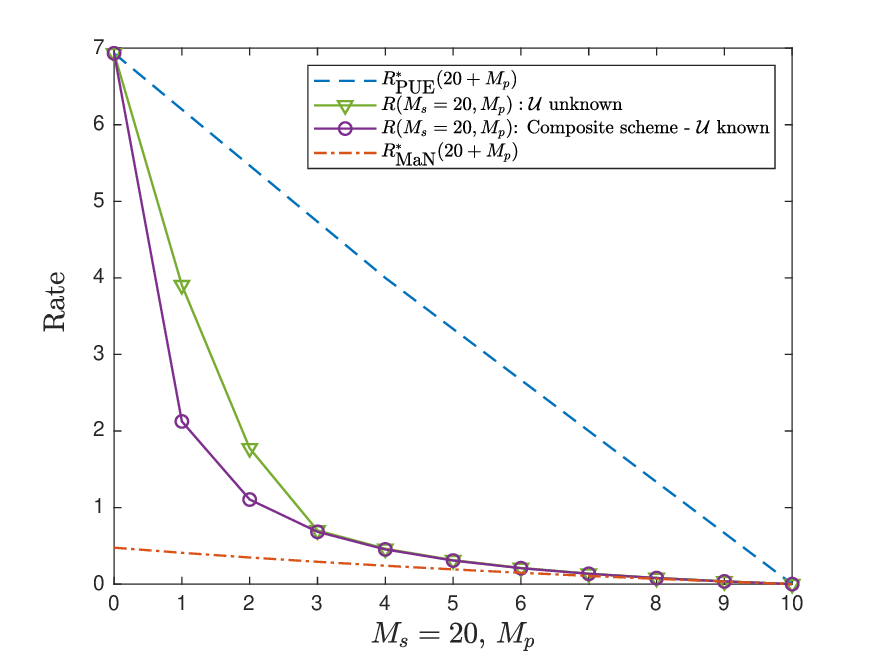}
		\caption{$M_s=20$}
		\label{fig:ts2skewed}
	\end{subfigure}
	%	\hfill
	%	\begin{subfigure}{0.5\textwidth}
	%		\includegraphics[width=\textwidth]{ts=3,profile=skewed.eps}
	%    \caption{$M_s=15$}
	%    \label{fig:ts3skewed}
	%	\end{subfigure}
	\caption{\mbox{For a network with $K=30$, $N=30$, $\Lambda=5$, $\mathcal{U}=\{\{1,2,\ldots,20\}$,$\{21,\ldots,24\}$},$\{25,26,27\}$,$\{28,29\},\{30\}\}$. }
	\label{fig:skewed}
\end{figure*}

\begin{figure*}
	\centering
	\begin{subfigure}{0.45\textwidth}
		\includegraphics[width=\linewidth]{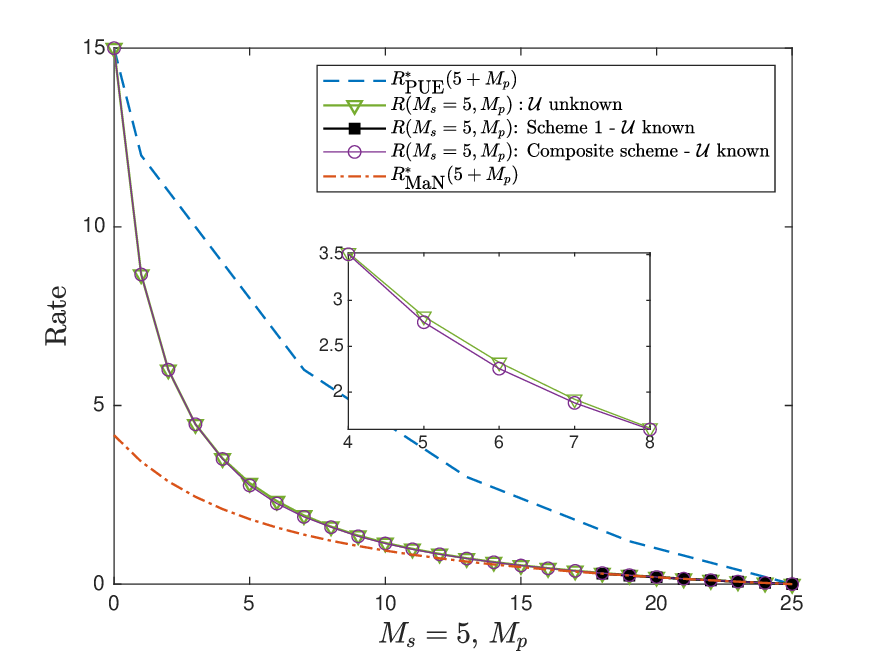}	
		\caption{$M_s=5$}
		\label{fig:ts1uniform}
	\end{subfigure}
	\hfill
	\begin{subfigure}{0.45\textwidth}
		\includegraphics[width=\linewidth]{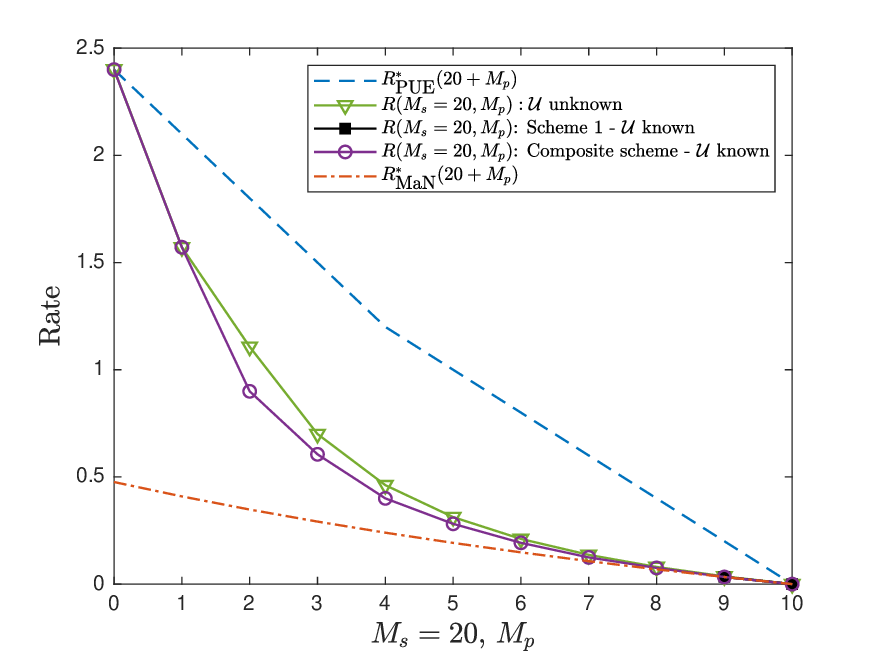}
		\caption{$M_s=20$}
		\label{fig:ts2uniform}
	\end{subfigure}
	%	\hfill
	%	\begin{subfigure}{0.45\textwidth}
	%		\includegraphics[width=\linewidth]{ts=3,profile=uniform.eps}
	%		\caption{$M_s=15$}
	%		\label{fig:ts3uniform}
	%		\hfill
	%	\end{subfigure}
	\caption{\mbox{For a network with $K=30$, $N=30$, $\Lambda=5$,  $\mathcal{U}=\{\{1,\ldots,6\}$,$\{7,\ldots,12\}$,}$\{13,\ldots,18\}$,$\{19,\ldots,24\},$ \\  $\textrm{\hspace{1cm}}\{25,\ldots,30\}\}$. }
	\label{fig:uniform}
\end{figure*}

\begin{figure*}
	\centering
\begin{subfigure}{0.45\textwidth}
	\includegraphics[width=\linewidth]{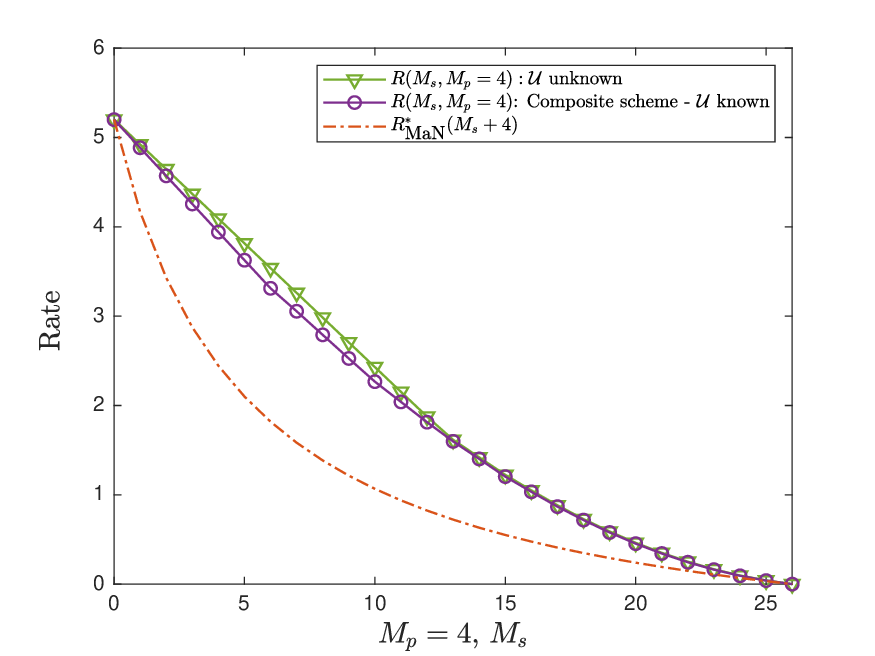}	
	\caption{$\mathcal{U}=\{\{1,2\ldots,20\},\{21,\ldots,24\},\{25.26,27\},$$\{28,29\},\{30\}\}$}
	\label{fig:Mpnonunif}
\end{subfigure}
\hfill
	\begin{subfigure}{0.45\textwidth}
	\includegraphics[width=\linewidth]{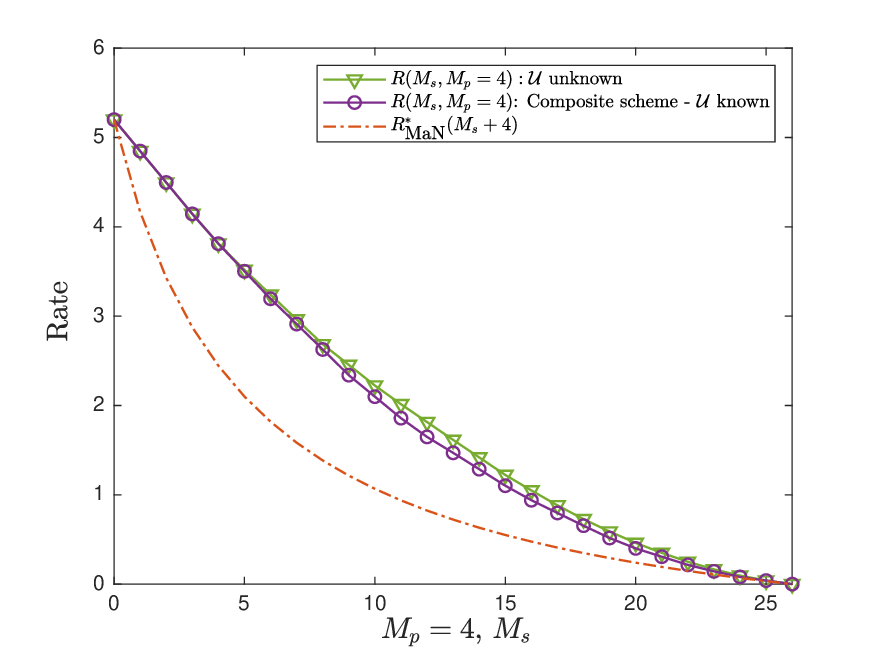}
	\caption{$\mathcal{U}=\{\{1,2\ldots,6\},\{7,\ldots,12\},\{13,\ldots,18\}$\\$ \textrm{\hspace{1.3cm}} \{19,\ldots,24\},\{25,\ldots,30\}\}$}
	\label{fig:Mpunif}
\end{subfigure}
	\caption{\mbox{For a network with $K=30$, $N=30$, $\Lambda=5$}}
\label{fig:Mpfixed}	
\end{figure*}

In Fig.~\ref{fig:ts1skewed} and Fig.~\ref{fig:ts2skewed}, all the plots are drawn by varying $M_p$ and fixing $M_s=5$ and $M_s = 20$, respectively. When $M_p = 0$, the schemes for both the $\mathcal{U}$ unknown and known cases achieve the optimal rate $R^{*}_{\textrm{PUE}}(M_s)$. For smaller values of $M_p$, there is a gain obtained by knowing $\mathcal{U}$ a priori in both the plots. However, due to the asymmetry in $\mathcal{U}$, we cannot exactly attribute that gain to any of the individual schemes (MaN scheme or Scheme $2$) in the composite scheme. At higher values of $M_p$, the improvement obtained by the composite scheme over the $\mathcal{U}$ unknown scheme is marginal. However, both  the proposed schemes' performances approach $R^{*}_{\textrm{MaN}}(M)$ in the higher memory regime of $M_p$. In Fig.~\ref{fig:ts1skewed}, Scheme 1 for $\mathcal{U}$ known case exists only when $24 \leq M_p \leq 25$. For all other values of $M_p$, the constraints $t=K(M_s+M_p)/N \geq \mathcal{L}_1$ and $\frac{\binom{K-\mathcal{L}_1}{t-\mathcal{L}_1}N}{\binom{K}{t}} \geq M_s$ are not satisfied jointly. Whereas, in the case of $M_s=20$, this condition is satisfied only at the trivial point $M_p = N-M_s=10$. Hence, it is not shown in Fig.~\ref{fig:ts2skewed}.

For $M_s=5$ and $M_s = 20$, the plots corresponding to the uniform user-to-cache association $\mathcal{U}=\{\{1,\ldots,6\},\{7,\ldots,12\}$,
$\{13,\ldots,18\},\{19,\ldots,24\},\{25,\ldots,30\}\}$ are given in Fig.~\ref{fig:ts1uniform} and Fig.~\ref{fig:ts2uniform}, respectively. When $M_s =5$, the composite scheme offers only a minuscule gain over the $\mathcal{U}$ agnostic scheme.
This is because the profile is uniform, hence the scheme in the $\mathcal{U}$ unknown case performs as good as the scheme in the known case presented in Section~\ref{subsec: profile_opt}. Whereas, in the case of $M_s = 20$, a significant gain is obtained by the composite scheme at the smaller values of $M_p$. In fact, the gain obtained from the composite scheme is attributed to the coding gain provided by Scheme $2$, which is $(t_s^{(S2)}+1)(t_p^{(S2)}+1)$. When $M_s = 5$, Scheme $1$ exists for $19 \leq M_p \leq 25$, and when $M_s = 20$, Scheme $1$ exists only for $9\leq M_p \leq 10 $.

In Fig.~\ref{fig:Mpnonunif} and Fig.~\ref{fig:Mpunif}, the rate-memory tradeoff of the proposed schemes are characterized by varying $M_s$ and fixing $M_p=4$ for a non-uniform $\mathcal{U}$ and a uniform $\mathcal{U}$, respectively. For non-uniform $\mathcal{U}$, a small gain is obtained by the composite scheme in the smaller values of $M_s$ region. Whereas, in the case of uniform profile, the improvement provided by the composite scheme is significant when $M_s$ is slightly large. This gain is explicitly due to the $(t_s^{(S2)}+1)(t_p^{(S2)}+1)$ coding gain provided by Scheme $2$ in the composite scheme.

Next, we compare the performance of our schemes  with that of MaN scheme and the optimal shared cache scheme in \cite{PUE} by keeping the system memory same in all the three network models. In the network model shown in Fig.~\ref{fig:setting}, the total system memory is $\Lambda M_s +KM_p$. The rate versus system memory plots are given in Fig.~\ref{fig:Mconstant} for a uniform and a non-uniform profile. Note that different $(M_s,M_p)$ pairs can lead to the same total system memory, but, the rates achieved by our schemes vary for each $(M_s,M_p)$ pair. The plots of our proposed schemes  in Fig.~\ref{fig:Mnonunif} and Fig.~\ref{fig:Munif} are drawn by varying $M_p$ and fixing $M_s = 20$ (same as in Fig.\ref{fig:ts2skewed} and Fig.~\ref{fig:ts2uniform}). Corresponding to each $M_p \in [0,N-20]$, the rate $R^{*}_{\textrm{PUE}}(M)$ is calculated for $M = (\Lambda M_s +KM_p)/\Lambda=M_s+\frac{K}{\Lambda}M_p$, and $R^{*}_{\textrm{MaN}}(M)$ corresponds to $M = (\Lambda M_s +K M_p)/K = M_p + \frac{\Lambda}{K}M_s$ as  the total system memory needs to be constant in all the considered network models. The rate $R^{*}_{\textrm{PUE}}(M)$ becomes zero at smaller values of $M_p$ itself, and $R^{*}_{\textrm{MaN}}(M)$ never goes to zero in Fig.~\ref{fig:Mnonunif} and Fig.~\ref{fig:Munif}, as $M_p + \frac{\Lambda}{K}M_s$ is always less than $N$ for every $0 \leq M_p \leq N - M_s$.  When $\mathcal{U}$ is uniform, $R^{*}_{\textrm{PUE}}(M/\Lambda)$  and $R^{*}_{\textrm{MaN}}(M/K)$ will always serve as a lower bound and an upper bound on the performance of our schemes (see Fig.~\ref{fig:Munif}). It is expected due to the fact that the total memory accessed by an individual user in different network models is in the order: $\frac{\Lambda}{K}M_s +  M_p \leq M_s + M_p \leq M_s + \frac{K}{\Lambda}M_p$ . When $\mathcal{U}$ is non-uniform, the rate-memory curves of the proposed schemes are not always bounded between  $R^{*}_{\textrm{PUE}}(M/\Lambda)$ and $R^{*}_{\textrm{MaN}}(M/K)$ (see Fig.~\ref{fig:Mnonunif}). This is because, in a shared cache network, the multicasting opportunities are severely affected by the skewness in $\mathcal{U}$.

\begin{figure*}
	\centering
	\begin{subfigure}{0.45\textwidth}
		\includegraphics[width=\linewidth]{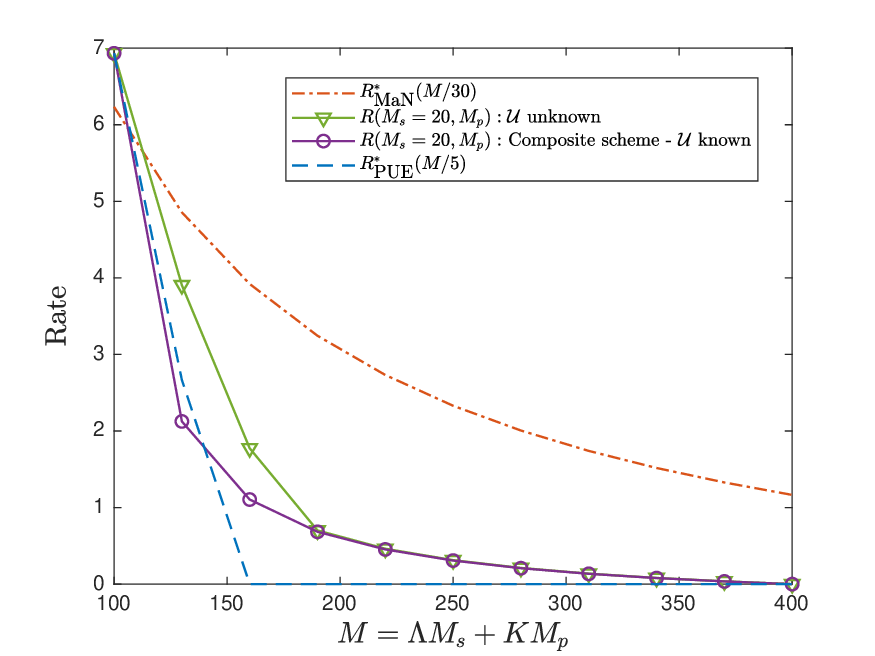}	
		\caption{$\mathcal{U}=\{\{1,2\ldots,20\},\{21,\ldots,24\},\{25.26,27\},\{28,29\},\{30\}\}$}
		\label{fig:Mnonunif}
	\end{subfigure}
	\hfill
	\begin{subfigure}{0.45\textwidth}
		\includegraphics[width=\linewidth]{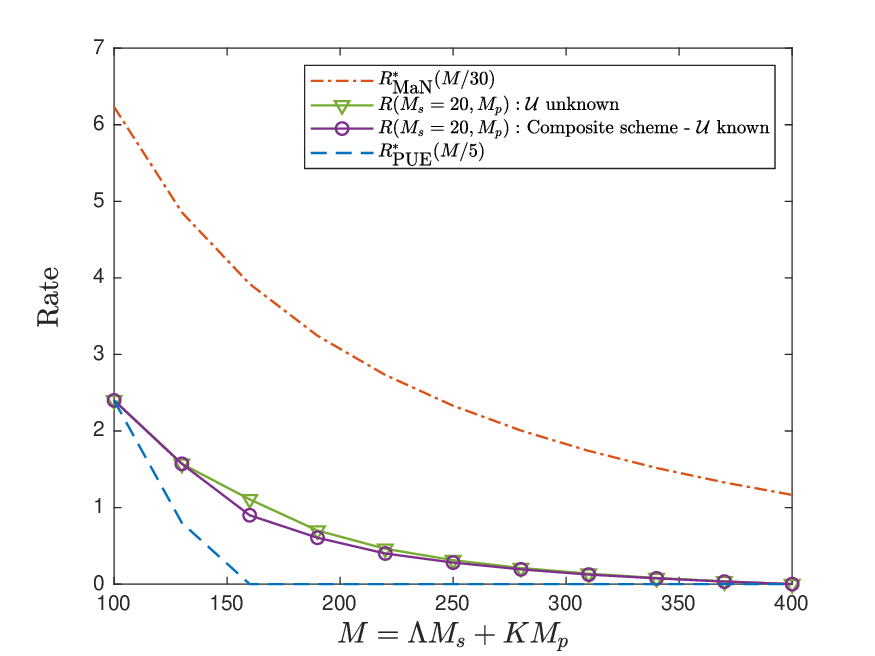}
		\caption{$\mathcal{U}=\{\{1,2\ldots,6\},\{7,\ldots,12\},\{13,\ldots,18\}$\\$ \textrm{\hspace{1.3cm}} \{19,\ldots,24\},\{25,\ldots,30\}\}$}
		\label{fig:Munif}
	\end{subfigure}
	\caption{\mbox{For a network with $K=30$, $N=30$, $\Lambda=5$}}
	\label{fig:Mconstant}	
\end{figure*}

\section{Lower Bound and Optimality of Scheme 2 in Certain Memory Regimes}
\label{sec:lb}
In this section, we derive a cut-set based lower bound on $R^{*}(M_s,M_p)$, and show that Scheme $2$ is optimal in certain memory regimes.

\begin{thm}
\label{thm4}
For $N$ files, $K$ users, each with a cache of size $M_p \geq 0$, assisted by $\Lambda$ helper caches, each of size $M_s \geq 0$, the following lower bound holds when $M_s+M_p \leq N$,
\begin{equation}
{R^{*}(M_s,M_p) \geq \underset{u \in \{1,2,\ldots,\min({N,K})\}}{\max} \left( u -\frac{uM_p+\lambda_uM_s}{\lfloor N/u \rfloor} \right)},
\label{eq:cutset2}
\end{equation}
where $\lambda_u$ represents the index of the helper cache to which the $u^{th}$ user is connected, assuming users and helper caches are ordered.
\end{thm}
\begin{IEEEproof}
Let $u \in [\min\{N,K\}]$, and consider the first $u$ users. Assuming that the caches and the users are ordered, let $\lambda_u$ denote the helper cache to which the $u^{th}$ user is connected. For a demand vector $\mathbf{d}_1=(1,2,\ldots,u,\phi,\ldots,\phi)$, where the first $u$ users request the files $W_1,W_2,\ldots,W_u$, respectively, and the remaining $K-u$ users' requests can be arbitrary, the sever makes a transmission $X_1$. The transmission $X_1$, the cache contents $\mathcal{Z}_{[1:\lambda_u]}$ and $Z_{[1:u]}$ enable the decoding of the files $W_{[1:u]}$. Similarly, for a demand vector $\mathbf{d}_2=(u+1,u+2,\ldots,2u,\phi,\ldots,\phi)$, the transmission $X_2$ and the cache contents $\mathcal{Z}_{[1:\lambda_u]}$, $Z_{[1:u]}$ help to recover the files $W_{[u+1:2u]}$. Therefore, if we consider $\lfloor N/u \rfloor$ such  demand vectors and its corresponding transmissions $X_{[1:\lfloor N/u \rfloor ]}$, the files $W_{[1:u\lfloor N/u \rfloor]}$ can be decoded by the first $u$ users using their available cache contents. Thus, we obtain
\begin{equation}
\begin{aligned}
u \lfloor N/u \rfloor  &\leq uM_p + \lambda_uM_s + \lfloor N/u \rfloor R^{*}(M_s,M_p).
\label{eq:cutset}
\end{aligned}
\end{equation}
On rearranging \eqref{eq:cutset}, we get
%\begin{equation}
$R^{*}(M_s,M_p) \geq u -\frac{uM_p+\lambda_uM_s}{\lfloor N/u \rfloor}$.
%\end{equation}
Optimizing over all possible choices of $u$, we obtain
\begin{equation*}
{R^{*}(M_s,M_p) \geq \underset{u \in \{1,2,\ldots,\min({N,K})\}}{\max} \left( u -\frac{uM_p+\lambda_uM_s}{\lfloor N/u \rfloor} \right)}.
%\label{eq:cutset2}
\end{equation*}
This completes the proof of Theorem~\ref{thm4}.
\end{IEEEproof}

The lower bound in \eqref{eq:cutset2} holds irrespective of the type of placement employed. Using the above lower bound, we now show that Scheme $2$ is optimal when $M_s \geq N(1-\frac{1}{\Lambda})$ and $N \geq M_s+M_p \geq N(1-\frac{1}{\Lambda \mathcal{L}_1})$.
\begin{lem}
\label{lem2}
	For $M_s \geq N(1-\frac{1}{\Lambda})$ and $N \geq M_s+M_p \geq N(1-\frac{1}{\Lambda\mathcal{L}_1})$, the optimal rate is
	\begin{equation}
	  R^{*}(M_s,M_p)= 1-\frac{M_s+M_p}{N}.
	\end{equation}
%\label{lem2}
\end{lem}
\begin{IEEEproof}
Consider Scheme $2$ discussed in Section~\ref{subsec:scheme2}. To satisfy the condition $M_s \geq N(1-\frac{1}{\Lambda})$, $t_s$ needs to be either $\Lambda-1$ or $\Lambda$. When $t_s=\Lambda-1$, $M_s=N(1-\frac{1}{\Lambda})$, and $t_s=\Lambda$ corresponds to $M_s=N$. 

Consider the case $t_s=\Lambda-1$. In order to hold the condition $M_s+M_p \geq N(1-\frac{1}{\Lambda \mathcal{L}_1})$, $M_p$ needs to be greater than $\frac{N}{\Lambda}(1-\frac{1}{\mathcal{L}_1})$. The condition $M_p \geq \frac{N}{\Lambda}(1-\frac{1}{\mathcal{L}_1})$ implies $t_p=\mathcal{L}_1-1$ or $t_p=\mathcal{L}_1$. Then, corresponding to the $(t_s,t_p)$ pairs $(\Lambda-1,\mathcal{L}_1-1)$ and $(\Lambda-1,\mathcal{L}_1)$, using Scheme $2$, the memory-rate triplets $(M_s,M_p,R)$: $\big(N(1-\frac{1}{\Lambda}),\frac{N}{\Lambda}(1-\frac{1}{\mathcal{L}_1}), \frac{1}{\Lambda\mathcal{L}_1}\big)$ and $\big(N\big(1-\frac{1}{\Lambda}\big),\frac{N}{\Lambda},0\big)$ are achievable, respectively. When $t_s=\Lambda$, we get $M_s=N$ and $M_p=0$, which means all the files are entirely cached in each helper cache, and the rate is $R(N,0)=0$. Thus, the memory-rate triplet $(N,0,0)$ is also achievable.
By memory-sharing among the triplets $\big(N(1-\frac{1}{\Lambda}),\frac{N}{\Lambda}(1-\frac{1}{\mathcal{L}_1}), \frac{1}{\Lambda\mathcal{L}_1}\big)$, $\big(N\big(1-\frac{1}{\Lambda}\big),\frac{N}{\Lambda},0\big)$, and $(N,0,0)$, the rate 
%\begin{equation}
$ R(M_s,M_p)=1-\frac{M_s+M_p}{N}$
% \label{eq:op1}
%\end{equation}
is achievable by Scheme $2$ when $M_s \geq N(1-\frac{1}{\Lambda})$ and $N \geq M_s+M_p \geq N(1-\frac{1}{\Lambda \mathcal{L}_1})$.

Next, consider the cut-set bound obtained in \eqref{eq:cutset2}. Substitute $u=1$, then we get:
%\begin{equation}
$R^{*}(M_s,M_p)\geq 1-\frac{M_s+M_p}{N}.$ 
%\label{eq:op2}
%\end{equation}
Thus, we obtain
$R^{*}(M_s,M_p)=1-\frac{M_s+M_p}{N}
$
for $M_s \geq N(1-\frac{1}{\Lambda})$ and $N \geq M_s+M_p \geq N(1-\frac{1}{\Lambda \mathcal{L}_1})$.
This completes the proof of Lemma~\ref{lem2}.
\end{IEEEproof}

\section{Conclusion}
\label{sec:concl}

In this work, we studied a network model where the users have access to two types of caches: a private cache and a cache shared with some of the users. For this network model, we proposed centralized coded caching schemes under uncoded placement for two different scenarios: with and without the prior knowledge of user-to-helper cache association. With user-to-helper cache association known at the server during the placement, we designed four schemes: one inspired from the optimal dedicated cache scheme and has limited operational regime, the two other proposed schemes exhibited improved performances in different memory regions, hence both are combined to form a fourth scheme that always result in a better performance than the scheme in the unknown case. The optimality of some of the schemes in the known case are also shown, either by matching the rate with the known optimality results or by matching it with the derived cut-set bound. Deriving stronger converse bounds and characterizing the achievability results in all memory regimes is a direction to work on further.

%with improved performances than the scheme in the unknown case. One scheme is shown to be optimal under uncoded placement. However, its applicability is limited. The second scheme, which exists for all memory points, is shown to be optimal in certain memory regimes by matching with the derived cut-set bound. It would be interesting to derive tighter lower bounds for this network model and characterize the performance of the schemes. Also, exploring schemes with coded placement is another direction to work on.
 
\section*{Acknowledgment}

This work was supported partly by the Science and Engineering Research Board (SERB) of Department of Science and Technology (DST), Government of India, through J.C Bose National Fellowship to Prof. B. Sundar Rajan, and by the Ministry of Human Resource Development (MHRD), Government of India through Prime Minister's Research Fellowship to Elizabath Peter and K. K. Krishnan Namboodiri.

\begin{appendices}
	\section{}
	\label{sec:appendix1}
 In this section, we find the value of $f$ that minimizes the function $\frac{f(K-\frac{KM_s}{fN})}{1+\frac{\Lambda M_s}{fN}}+\frac{(1-f)(K-\frac{KM_p}{(1-f)N})}{1+\frac{K M_p}{(1-f)N}}$, where $M_s/N\leq f\leq 1-M_p/N$. We have
\begin{equation*}
f^\ast = \arg\min_{f} \frac{f(K-\frac{KM_s}{fN})}{1+\frac{\Lambda M_s}{fN}}+\frac{(1-f)(K-\frac{KM_p}{(1-f)N})}{1+\frac{K M_p}{(1-f)N}}.
\end{equation*} 
To find $f^\ast$, we equate the first derivative of the function to zero, and we obtain the following:
\begin{align}
& \diff{}{f}\left(\frac{f(1-\frac{M_s}{fN})}{1+\frac{\Lambda M_s}{fN}}+\frac{(1-f)(1-\frac{M_p}{(1-f)N})}{1+\frac{K M_p}{(1-f)N}}\right) =0,\notag\\
\implies & \diff{}{f}\left(\frac{f-\frac{M_s}{N}}{1+\frac{\Lambda M_s}{fN}}\right) =-\diff{}{f}\left(\frac{1-f-\frac{M_p}{N})}{1+\frac{K M_p}{(1-f)N}}\right). 
\label{eq:der_zero}
\end{align}
For ease of representation, we denote $m_p =M_p/N$ and $m_s = M_s/N$. Then, \eqref{eq:der_zero} becomes
\begin{align}
\frac{1+\frac{\Lambda m_s}{f}+(f-m_s)\frac{\Lambda m_s}{f^2}}{(1+\frac{\Lambda m_s}{f})^2} = \frac{1+\frac{K m_p}{1-f}+(1-f-m_p)\frac{K m_p}{(1-f)^2}}{(1+\frac{K m_p}{1-f})^2}.
\label{eq:der_zero1}
\end{align}
Rearranging the terms and completing squares on both sides in \eqref{eq:der_zero1} give
\begin{align*}
\frac{(1+\frac{\Lambda m_s}{f})^2-\frac{\Lambda(\Lambda+1) m_s^2}{f^2}}{(1+\frac{\Lambda m_s}{f})^2} = \frac{(1+\frac{K m_p}{1-f})^2-\frac{K(K+1) m_p^2}{(1-f)^2}}{	(1+\frac{K m_p}{1-f})^2}.
\end{align*}
Therefore, we have
\begin{align*}
1-\frac{\Lambda(\Lambda+1) m_s^2}{(f+\Lambda m_s)^2} =1-\frac{K(K+1) m_p^2}{(1-f+Km_p)^2}.
\end{align*}
By further rearranging, we get
%\begin{small}
\begin{align*}
\frac{ (f+\Lambda m_s)^2}{(1-f+Km_p)^2} = \frac{\Lambda(\Lambda+1) m_s^2}{K(K+1) m_p^2}.
 \end{align*}
%\end{small}
Upon further reduction, we obtain
 \begin{small}
 \begin{align*}
& (f+\Lambda m_s)\sqrt{K(K+1)} m_p    =  (1-f+Km_p)\sqrt{\Lambda(\Lambda+1)} m_s, \\
&\implies f(\sqrt{K(K+1)} m_p+ \sqrt{\Lambda(\Lambda+1)} m_s)   = \\ &\hspace{1.8cm}(1+Km_p)\sqrt{\Lambda(\Lambda+1)} m_s -\Lambda m_s\sqrt{K(K+1)} m_p.
\end{align*}
\end{small}
%and subsequently
%\begin{align*}
%(f+\Lambda m_s)\sqrt{K(K+1)} m_p    =  (1-f+Km_p)\sqrt{\Lambda(\Lambda+1)} m_s   ,
%\end{align*}
%and
%\begin{align*}
%f(\sqrt{K(K+1)} m_p+ \sqrt{\Lambda(\Lambda+1)} m_s)   =  (1+Km_p)\sqrt{\Lambda(\Lambda+1)} m_s -\Lambda m_s\sqrt{K(K+1)} m_p.
%\end{align*}
Thus, we obtain
\begin{align}
f = \frac{M_s}{N}\frac{(1+\frac{KM_p}{N})\sqrt{\Lambda(\Lambda+1)} - \frac{\Lambda M_p}{N}\sqrt{K(K+1)}}{\sqrt{K(K+1)}\frac{M_p}{N}+ \sqrt{\Lambda(\Lambda+1)}\frac{M_s}{N}}.
\label{eq:fexp}
\end{align}
Since  $f\geq M_s/N$, \eqref{eq:fexp} can be rewritten as
\begin{small}
\begin{equation}
\label{fopt}
f^\ast = \frac{M_s}{N}\max\left(1,\frac{(1+\frac{KM_p}{N})\sqrt{\Lambda(\Lambda+1)} - \frac{\Lambda M_p}{N}\sqrt{K(K+1)}}{\sqrt{K(K+1)}\frac{M_p}{N}+ \sqrt{\Lambda(\Lambda+1)}\frac{M_s}{N}}\right).
\end{equation}
\end{small}

\end{appendices}

\end{document}